\documentclass[journal]{IEEEtran}
\usepackage{times}
\usepackage[T1]{fontenc}
\usepackage[greek,english]{babel}
\usepackage[cmex10]{amsmath}
\usepackage{graphicx}
\usepackage{enumerate}
\usepackage{caption}
\usepackage{color}
\usepackage{subfigure}
\usepackage{float}
\usepackage[noadjust]{cite}
\usepackage{soul}
\usepackage[justification=centering]{caption}
\usepackage{epstopdf}
\usepackage{lipsum}
\usepackage{amsmath,amssymb}
\usepackage{xcolor}
\usepackage{epstopdf}
\usepackage{amsthm}
\usepackage{indentfirst}
\usepackage{adjustbox}

\usepackage[utf8]{inputenc}
\usepackage[english]{babel}
 
\newtheorem{theorem}{Theorem}[section]

\newtheorem{lemma}{Lemma}

\begin{document}

\title{Delay Optimal Scheduling for Chunked Random Linear Network Coding Broadcast} 

\author{Emmanouil Skevakis, Ioannis Lambadaris, Hassan Halabian\\
\IEEEauthorblockA{Department of Systems and Computer Engineering, Carleton University, Ottawa, Ontario, Canada\\
\textit{email: \{eskevakis, ioannis, hassanh\}@sce.carleton.ca}}}

\maketitle

\begin{abstract}
We study the broadcast transmission of a single file to an arbitrary number of receivers using Random Linear Network Coding (RLNC) in a network with unreliable channels. Due to the increased computational complexity of the decoding process (especially for large files) we apply \textit{chunked} RLNC (i.e. RLNC is applied within non-overlapping subsets of the file). 

In our work we show the \textit{optimality} of the \textit{Least Received (LR)} batch scheduling policy (which was introduced in our prior work) with regards to the expected file transfer completion time. Furthermore, we refine some of our earlier results, namely the expected file transfer completion time of the LR policy and the minimum achievable coding window size in the case of a user defined delay constraint. Finally, we experimentally evaluate a modification of the LR policy in a more realistic system setting with reduced feedback from the receivers.

\end{abstract}

\section{Introduction}

Due to the constantly increasing demand in multimedia traffic, improving the performance of \textit{broadcast} and \textit{multicast} communications is of major importance. Network coding (NC) is widely studied in such transmission systems. With NC, the original data packets are encoded (based on an encoding scheme) and transmitted. Recent studies revealed that network coding can enhance the performance of wireless networks in broadcast and multicast transmissions (\cite{eryilmaz2008delay, yang2012throughput, xie2013network, katti2006xors, ghaderi2008reliability, ho2006random}). It has been shown that network coding can provide significant gains in terms of transmission delay (\cite{eryilmaz2008delay, xie2013network}), achievable throughput (\cite{yang2012throughput}) and overall reliability (\cite{ghaderi2008reliability}) of the underlying network (e.g. network coding can increase the multicast and broadcast capacity of unreliable links (\cite{katti2006xors,ghaderi2008reliability})). 

Random linear network coding (RLNC) is a simple, yet efficient, encoding scheme (\cite{ho2006random}). Under RLNC, the encoded packets are created by linearly combining a predefined number of $K$ packets (known as the \textit{coding window size}) prior to the transmission. Upon successful reception of $K$ linearly independent encoded packets, a receiver will be able to decode them (e.g. using Gaussian elimination). In broadcast and multicast systems, RLNC has been shown not only to enhance the average throughput and delay (or completion time) \cite{katti2006xors, ghaderi2008reliability, ho2006random, koller2011optimal, eryilmaz2008delay} but also to approach the system capacity with negligible overhead \cite{ho2004networking}. However, RLNC may result in increased decoding delay (the entire coding window of $K$ packets needs to be received/stored before the beginning of the decoding process) and storage/computational requirements (the complexity of the decoding process is $O(K^3)$). Chunked network coding has been proposed (\cite{silva2009sparse, joshi2013round}) in order to reduce the computational complexity of RLNC. In chunked network coding the message is divided into chunks (or generations, blocks) of packets and the encoding scheme (usually RLNC) is applied to each chunk.

In \cite{skevakis2016decoding}, we studied the one hop broadcast transmission of a single file to an arbitrary number of receivers in an unreliable (wireless) network, using chunked RLNC. In our system, the receivers reject encoded packets of future chunks until all such previous chunks are succesfully decoded. The benefits of such an approach are the following. First, receivers are relieved of increased storage requirements; only $K$ packets need to be stored. Second, ordering the delivery of packet chunks is prefered in mulitmedia streaming applications such as YouTube and Netflix. In \cite{skevakis2016decoding} we developed and evaluated a scheduling policy, namely the \textit{Least Received (LR)}, when $K$ is less than the file size. We showed that near optimal completion time\footnote{The minimum (optimal) file transfer completion time will be achieved when the coding window size equals the file size (\cite{joshi2013round}).} can be achieved with small values of $K$ (smaller $K$ results in timely delivery of earlier packets), under the $LR$ policy. We dervived closed form approximation formulas for a) the expected file transfer completion time and b) the minimum achievable coding window size $K$ given a user defined delay constraint. This constraint is expressed in the form of the relative increase with regards to the optimal completion time (\cite{joshi2013round}). In \cite{skevakis2016optimal}, we proved the optimality, in terms of minimizing the file transfer completion time, of our proposed $LR$ policy, using Dynamic Programming for the \textit{special case} of two receivers. 

\subsection{Contributions and Related Work}

The \textit{contributions} of this paper are:
\begin{enumerate}
\item Proof of the optimality of the $LR$ policy with respect to the file transfer time for an arbitrary number of receivers.
\item Derivation of further approximations for the expected file transfer completion time and the minimum achievable coding window size $K$ given a user defined delay constraint (\cite{skevakis2016decoding}) and evaluation of their accuracy. 
\item Simulation comparison and performance assessment of our proposed LR policy with other policy heuristics.
\item Proposed extension of the $LR$ policy for more realistic systems with limited feedback.
\end{enumerate}
\noindent Similar studies, but not with the same objectives, have already been performed. Such studies can be divided into two groups, based on the adopted research direction. In the first group, general performance properties of RLNC based on the transmission of a single chunk of $K$ packets (and not of the entire file) are investigated. Eryilmaz et al. \cite{eryilmaz2008delay} quantified the throughput and delay gains of network coding when compared to traditional transmission strategies. In \cite{xie2013network}, tight bounds for the expected delay per packet under uncoded transmissions are derived and compared with the expected delay per packet under RLNC. In \cite{swapna2013throughput} it is shown that if the coding window size scales with the number of receivers, the throughput will converge to the broadcast capacity. Yang et al. \cite{yang2012throughput} focused on a similar problem (system throughput as a function of the number of receivers) in time correlated erasure channels. 

In the second, the researchers focused on the benefits of chucked RLNC. The authors of \cite{koller2011optimal} investigated the optimal block size in order to minimize the expected number of transmissions. In \cite{koller2011optimal}, RLNC is applied over the blocks (an encoded block is a linear combination of all the blocks) and not within each block, as in our work.\footnote{The benefits of our approcah were discussed earlier.} In \cite{silva2009sparse} and \cite{joshi2013round} the authors focused on reducing the computational complexity of the decoding processes using chunked network coding with overlapping classes. Their work mainly applies to unicast and not broadcast sessions.

The remainder of this paper is organized as follows: Our system model along with the neccessary notation is introduced in section \ref{SystemModel}. In section \ref{proof} we provide the proof for the optimality of the $LR$ policy. Section \ref{approx} contains approximations (extending the ones we presented in \cite{skevakis2016decoding}) for the expected file transfer completion time and the minimum achievable coding window size $K$ given a user defined delay constraint, under chunked RLNC. Our experimental results are presented in section \ref{exps} along with a brief extension of our proposed policy in the case of limited feedback. We conclude and present further research suggestions in section \ref{conclusions}.

\section{Problem Statement - Main Results}\label{SystemModel}
\subsection{System Model}
We consider the wireless, one-hop transmission of a single file to an arbitrary number of receivers ($N$) using network coding. Despite its simplicity, our system model can capture the characteristics of current cellular and satellite systems and it may be used in the analysis of more complex network topologies. Our model is described below.

\textit{Base Station:} The base station holds a single file that contains $F$ packets. The file is (virtually) divided into consecutive and non-overlapping subsets (batches) of $K$ packets. Therefore, batch $i$ contains the packets $(i-1)K$ to $iK-1$. 
We let $b = \frac{F}{K}$\footnote{For the purpose of this study, we assume $\frac{F}{K}$ to be an integer.} denote the total number of batches. The base station transmits encoded packets, where each encoded packet is generated by applying RLNC within a batch. At each time slot, the base station selects a batch based on a scheduling policy and the channel state information (CSI); we assume that the base station has complete knowledge of the connected receivers. The assumption of perfect CSI can be impractical, especially when the number of receivers is large. However, by studying such a scenario, we can derive the optimal actions (that lead to the minimum file transfer completion time) in the ideal case. This can give us useful insights for the optimal actions in the case of limited CSI and provide strict lower bounds on the completion time in such cases. 
%

\textit{Channel:} The channels connecting the receivers and the base station are modelled by i.i.d ON/OFF processes. The state of each channel is represented by a Bernoulli r.v. with mean $p$. Time is slotted and only one packet can be transmitted at each time slot. If the corresponding channel is ON the packet is received with no errors.

\textit{Receivers:} The state $X_i(t)$ of each receiver $r_i$ is a random variable representing the total number of received packets at time $t$. Each receiver has a buffer where the received encoded packets are stored. Upon successful reception of $K$ encoded packets (of the same batch), the packets are decoded and removed from the queue. We assume linear independence of the encoded packets and negligible coding overhead (coefficients for the linear combinations), as presented in \cite{eryilmaz2008delay}. Each receiver has an attribute, namely the batch ID. The batch ID of receiver $i$ at $t$ is $\lfloor \frac{X_i(t)}{K} \rfloor +1$ and represents the batch of the encoded packet(s) that $r_i$ is expecting. Any out of order packets (encoded packets of batch $i$ received by a receiver with batch ID $j \neq i$) are discarded by the receiver.

\textit{Batch Scheduling Policy:} When applying chunked RLNC, a scheduling policy must be defined in order to select a single batch to be encoded (and thus transmitted) at each time slot. Figure \ref{System} shows an example of a system at some time $t$. Receivers R1 and R2 have successfully all of the first $K$ packets and are thus expecting encoded packets of the second batch. Recever R3 has received 2 packets and is expecting an encoded packet of the first batch. At some time slots, hereafter referred to as \textit{conflict slots}, a policy has more than one candidate batches (e.g. figure \ref{System}, assuming that all of the receivers are connected). Any policy which selects a "useful" batch to encode (a batch that will successfully be received by at least one receiver) at each time slot where at least one of the receivers is connected (i.e. no idling whenever possible), is referred to as a \textit{feasible policy}.

\subsection{Problem Statement}
The goal of our study is to find and evaluate the optimal policy in order to minimize the file transfer completion time (i.e. until \textit{all} receivers receive the entire file). 

Our proposed policy is the \textit{Least Receiver (LR)} batch scheduling policy (\cite{skevakis2016decoding}). $LR$ selects, at each time slot, the useful batch with the minimum ID (batch 1 in figure \ref{System}). 
The selection of the LR policy as a candidate for minimizing the expected file transfer time is based on the following intuition.

\textit{1) Finite file size:} We transmit a single file of finite length and we aim to minimize its total transfer time to the receivers. The file transfer completion time will be determined by the receivers which will be the last to receive the file (any receiver which "finished" earlier will not contribute to the total transfer time). Thus, intuitively, by favouring the receivers with the least number of received packets, we expect that the file transfer completion time will be decreased.

\textit{2) Queue balancing:} The LR policy, by favouring the receivers with the smallest batch ID, decreases the differences (i.e. spreading) of the batch IDs among the receivers. Therefore, the probability of having a large number of receivers with the same batch ID is increased. Thus, in the long run, each transmitted packet should be beneficial to more receivers.

\begin{figure}
\centering
\includegraphics[scale = 0.6, trim = 2 2 2 2, clip]{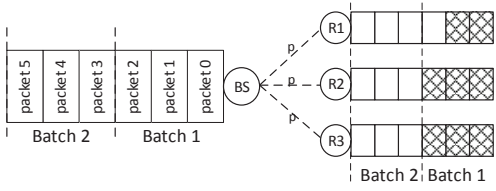}
\caption{System at time $t$. $K=3$, $N=3$.\\ ($R2$ and $R3$ have successfully received batch 1)}
\label{System}
\end{figure}

\subsection{Main Result}
In section \ref{proof} we will prove the following theorem, which formally shows the optimality of the $LR$ policy.
\begin{theorem}
Let $T^{\pi}$ denote the file transfer completion time under a batch scheduling policy $\pi$. Given any feasible policy $\pi^{(0)} \in \Pi$, where $\Pi$ is the set of all feasible policies, we can construct a sequence of policies $\pi^{(n)}$ such that :
\begin{center}
$T^{\pi^{(n)}} \leq T^{\pi^{(n-1)}}, n = 1, ..., 2b+1$ \\(where $b = \frac{F}{K}$, the number of batches in the original file)
\end{center}
The final policy ($\pi^{(2b+1)}$) will be the $LR$ policy.
\label{thm}
\end{theorem}

\section{Proof of theorem \ref{thm}}\label{proof}

\subsection{Notation - Definitions}

The notation used in section \ref{proof} will now be introduced.
\begin{itemize}
\item $X_i^{\pi}(t)$ denotes the r.v. representing the total received packets of receiver $i$ at time $t$ based on policy $\pi$.
\item $T_i^{\pi}$ denotes the completion time of receiver $i$ based on policy $\pi$ ($T_i^{\pi} = \min \{t : X_i^{\pi}(t) = F \}$).
\item $T^{\pi}$ denotes the file transfer completion time based on policy $\pi$ ($T^{\pi} = \max \{T_i^{\pi}\}$).
\item $R_{\beta}(t)$ denotes the set of receivers that are expecting encoded packets of batch $\beta$ at time $t$, $\beta = 1,.., b, F$ ($R_F$ contains the receivers that have received the entire file).
\item $R_{\beta}^*(t)$ denotes the set of receivers that require only one packet to decode batch $\beta$ at time $t$, $\beta = 1,.., b$.
\item The \textit{bottleneck receiver(s)} are the receiver(s) that have correctly received the smallest number of packets.
\item $D^{\pi}(t) = \beta$ : The decision of policy $\pi$ at $t$ is to transmit an encoded packet generated from bath $\beta$. 

\end{itemize}

For the proof, we will employ the tools of sample path analysis and stochastic dominance. Sample path methods have been widely used in the performance analysis and control of queueing systems. Loosely speaking, sample paths from a queueing system under two different operational regimes (e.g., control policies) are grouped and compared pointwise in an attempt to prove that one sample path \textit{dominates} the other, thus leading to performance comparisons. The technique is also known as sample path coupling and proves \textit{stochastic ordering} for the random variables of interest. The reader is referred to \cite{liu1995sample, tassiulas1993dynamic} for further details. In our work, we proceed with backwards induction. We define certain states of the system ($S_{b-1}, S_{b-2}, ...$) that will be visited by any scheduling policy. In particular, state $S_{b-1}$ (figure \ref{state01}) is defined as the one where all bottleneck receivers are expecting one packet for successfully decoding batch $b-1$, $S_{b-2}$ (figure \ref{state2}) is the state where the bottleneck receivers are short of one packet to complete batch $b-2$ and states $S_{b-3}, S_{b-4}, ..., S_{1}$ are defined similarly. For each state we will examine two \textit{milestone} time slots, depending on the connectivities of the bottleneck receivers. As an example, for state $S_{b-1}$, we will examine the system at times $t_{b-1}$ and $t_{b-1}^{ON}$, defined as follows; $t_{b-1}$, is the first time slot where the system reaches state $S_{b-1}$. $t_{b-1}^{ON}$ is the first time slot where the system arrives at state $S_{b-1}$ \textit{and} all bottleneck receivers are connected. The milestone time slots for $S_{b-2}. S_{b-3}, ...$ are defined in a similar fashion.

\subsection{Roadmap of the proof}

\begin{figure}
\centering
\includegraphics[trim = {2 2 2 2}, clip, scale = 0.45]{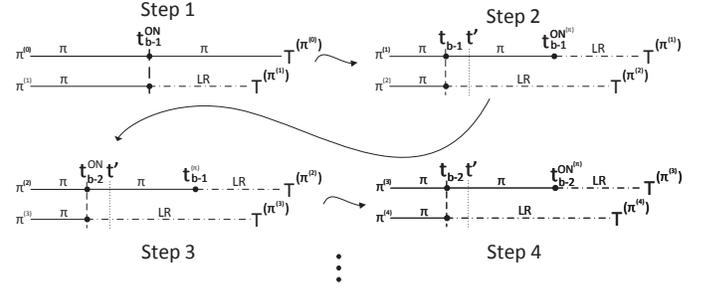}
\caption{Sketch of the proof of theorem \ref{thm}}
\label{steps}
\end{figure}

Starting from the latest milestone, $t_{b-1}^{ON}$, we show that the completion time of any policy $\pi$ can be improved by switching to the $LR$ policy at that time slot (figure \ref{steps}, Step 1). Then, based on the result of the previous step, we show that the completion time of any policy can be decreased if we follow the $LR$ policy from $t_{b-1}$ and onwards (figure \ref{steps}, Step 2). The rest of the milestone time slots are examined similarly (figure \ref{steps} - Step 3,4), by applying induction until the initial time $t=0$.

For the rest of this section, when comparing the queues of the receivers between two policies, we assume "coupled" connectivity vectors. The reader is suggested to refer to \cite{liu1995sample} for details on "sampe path coupling arguments".

\subsection{Theorem \ref{thm} - Details of the proof}

Suppose that our model operates under an arbitrary batch scheduling policy $\pi$. Such a system will visit (almost surely) state $S_{b-1}$ as depicted in figure \ref{state01}. As we mentioned previously, in $S_{b-1}$ the bottleneck receivers will be missing ONE packet to complete the reception of batch $b-1$. $R_{b-1}^*(t)$ is the set of these bottleneck receivers at time $t$. The rest of the receivers are either in batch $b$ (set $R_b(t)$) or have received the entire file (set $R_F(t)$). We assume that the set $R_b(t)$ is non-empty (if it is, the analysis is trivial). As we mentioned in section \ref{SystemModel}, we will examine two time slots while the system is in $S_{b-1}$, namely $t_{b-1}$ and $t_{b-1}^{ON}$. 

\begin{figure}
\centering
\includegraphics[scale = 0.45]{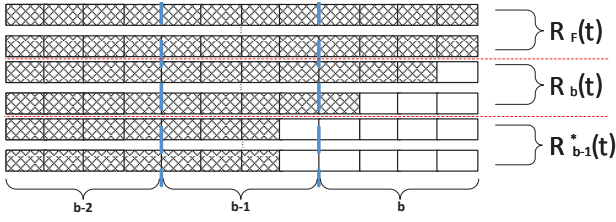}
\caption{State $S_{b-1}$, $K=4$}
\label{state01}
\end{figure}

\subsubsection{\textbf{\underline{Milestone time slot $t_{b-1}^{ON}$}}}\label{t_b-1ON}

Let $\pi^{(0)} = \pi$ and $\pi^{(1)}$ be a policy that makes the same decisions as $\pi$ until $t_{b-1}^{ON}$ and at $t_{b-1}^{ON}$ it switches to the $LR$ policy (figure \ref{steps}, Step 1). In this section, we will show that $T^{(\pi^{(0)})} > T^{(\pi^{(1)})}$. Hereinafter, we will refer to policy $\pi^{(0)}$ and $\pi^{(1)}$ as $\pi$ and $LR$, respectively. 
We assume that the time slot $t_{b-1}^{ON}$ is a conflict slot\footnote{If it is not, then all of the bottleneck receivers $r_i$ will receive a packet (regardless of the policy) and will move to batch $b$. As a result the completion time will be the same for every policy. This also happens if the decision of policy $\pi$, $D^{\pi}(t_{b-1}^{ON}) = D^{LR}(t_{b-1}^{ON}) = b-1$.}. At that time slot, $D^{\pi}(t_{b-1}^{ON}) = b$ and $D^{LR}(t_{b-1}^{ON}) = b-1$. Thus, 
\begin{center}
$X_i^{LR}(t_{b-1}^{ON} + 1) > X_i^{\pi}(t_{b-1}^{ON} + 1) \hspace{2 em} \forall i : r_i \in R_{b-1}^*(t_{b-1}^{ON})$
\end{center}

At $t_{b-1}^{ON}+1$, some receivers of the set $R_b(t_{b-1}^{ON})$ might move to $R_F(t_{b-1}^{ON}+1)$ with policy $\pi$ and \textit{all} of the receivers of set $R_{b-1}^*(t_{b-1}^{ON})$ will move to set $R_b(t_{b-1}^{ON}+1)$ with policy $LR$.
%
Since, under $LR$, all of the receivers are either expecting packets from batch $b$ or are finished, every time that a receiver $r_i$ is ON\footnote{By "ON" (or "OFF") we mean that a receiver is connected (or disconnected) to the base station.}, $r_i$ will be receiving a packet. Policy $\pi$ will probably have more conflict slots. Therefore, \textit{each time that a receiver $r_i$ in ON (for $t \geq t_{b-1}^{ON} + 1$), $r_i$ might receive a packet with $\pi$ but it will surely receive a packet with $LR$}. Thus, for $ t \geq t_{b-1}^{ON}$
\begin{equation}
X_i^{LR}(t+1) > X_i^{\pi}(t+1), \forall i : r_i \in R_{b-1}^*(t_{b-1}^{ON}) \cap R_b^{(LR)}(t)\footnote{{This equations holds for all $r_i$'s that were in $R_{b-1}^*(t_{b-1}^{ON})$ and are now in $R_b^{(LR)}(t)$; i.e. we are excluding the receivers that are finished at $t$.}} 
\label{1}
\end{equation}

Since each receiver $r_i$ will have received more packets with policy $LR$ and the completion time of receiver $r_i$ is $T_i^{\pi} = \min \{t : X_i^{\pi}(t) = F \}$ it follows that : 
\begin{equation}
T_i^{LR} < T_i^{\pi} \hspace{1 em} \forall i : r_i \in R_{b-1}^*(t_{b-1}^{ON})
\label{2}
\end{equation}

\begin{figure}
\centering
\includegraphics[scale = 0.65]{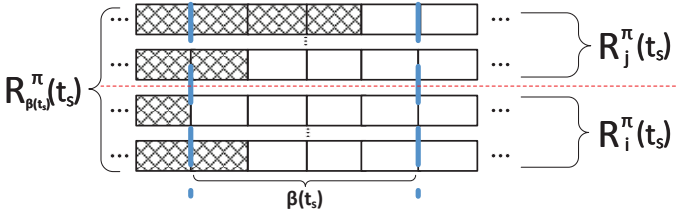}
\caption{Bottleneck receivers at $t_s$ (Lemma \ref{lemmaA})}
\label{lemA}
\end{figure}

\begin{lemma} Let, at $t_s$, $\beta(t_s)$ be the batch ID of the bottleneck receivers. Let $R_{\beta(t_s)}^{\pi}(t_s)$ be the set of such receivers, under any policy $\pi$, and $R_{i}^{\pi}(t_s)$ and $R_{j}^{\pi}(t_s)$ be two disjoint subsets of $R_{\beta(t_s)}^{\pi}(t_s)$ (figure \ref{lemA}). If, at $t_s$,
\begin{center}
$X_i^{\pi}(t_s) \leq X_j^{\pi}(t_s) \hspace{2 em} \forall i,j : r_i \in R_{i}^{\pi}(t_s), r_j \in R_{j}^{\pi}(t_s)$
\end{center}
then it holds that : 
\begin{center}
$\min\limits_{i} X_{i}^{\pi}(t+1) \leq \min\limits_{j} X_{j}^{\pi}(t+1) \hspace{2 em} t \geq t_s,$\\
for all the receiver indices $i$ and $j$ as above that remain in the set $R_{\beta(t_s)}^{\pi}$ at time $t > t_s$.\\ (i.e.,$i,j : r_{i} \in R_{i}^{\pi}(t_s) \cap R_{\beta(t_s)}^{\pi}(t)$
$, r_{j} \in R_{j}^{\pi}(t_s) \cap R_{\beta(t_s)}^{\pi}(t)$)
\end{center}
\label{lemmaA}
\end{lemma}
\begin{proof}[Proof of Lemma \ref{lemmaA} :] $ $\newline
The proof can be found in Appendix A. Furthermore, this lemma is general and will be also applied later.
\end{proof}

\noindent Using lemma \ref{lemmaA} for policy $LR$, with:
\begin{itemize}
\item $t_s = t_{b-1}^{ON} + 1$. At that time, under $LR$, all of the receivers are in the last batch $b$. Hence, $\beta(t_s) = b$.
\item $R_{\beta(t_s)}^{\pi}(t_s) = R_b^{LR}(t_{b-1}^{ON} + 1)$
\item $R_{i}^{\pi}(t_s) = R_{b-1}^*(t_{b-1}^{ON}) \cap R_{b}^{LR}(t_{b-1}^{ON}+1) = R_{b-1}^*(t_{b-1}^{ON})$. 
\item $R_{j}^{\pi}(t_s) = R_{b}(t_{b-1}^{ON}) \cap R_{b}^{LR}(t_{b-1}^{ON}+1) = R_{b}(t_{b-1}^{ON})$. 
\end{itemize}
we can show that, under $LR$, for $i$ : $r_i \in R_{b-1}^*(t_{b-1}^{ON})$ and $j$ : $r_j \in R_{b}(t_{b-1}^{ON})$
\begin{center}
$\min\limits_{i} X_{i}^{LR}(t) \leq \min\limits_{j} X_{j}^{LR}(t) \hspace{2 em} t \geq t_{b-1}^{ON} \Rightarrow$
\end{center}
\begin{equation}
\max\limits_{i} T_{i}^{LR} \geq \max\limits_{j} T_{j}^{LR}
\label{lemmaResult1}
\end{equation}
We know that the completion time of any policy $\pi'$ is :
\begin{equation}
T^{\pi'} = \max \{\max\limits_i T_i^{\pi'}, \max\limits_j T_j^{\pi'} \}
\label{3}
\end{equation}
Therefore, with the aid of eq. \ref{lemmaResult1} and \ref{2}, we can see that
\begin{equation}
\begin{split}
T^{LR} \stackrel{(\ref*{3})}{=} \max \{\max\limits_i T_i^{LR}, \max\limits_j T_j^{LR} \} \stackrel{(\ref*{lemmaResult1})}{=} \max\limits_i T_i^{LR} \stackrel{(\ref*{2})}{<} \\ \max\limits_i T_i^{\pi} \leq \max \{\max\limits_i T_i^{\pi}, \max\limits_j T_j^{\pi} \} \stackrel{(\ref*{3})}{=} T^{\pi} \Rightarrow \hspace{1.5em} \\ T^{(\pi^{(1)})} < T^{(\pi^{(0)})} \hspace{8 em} 
\end{split}
\label{5}
\end{equation}

Therefore, the delay optimal decision at time slot $t_{b-1}^{ON}$ is that of policy $LR$.
\subsubsection{\textbf{\underline{Milestone time slot $t_{b-1}$}}} \label{t_b-1}

In this part, we show that further reduction on the completion time of policy $\pi^{(1)}$ can be achieved if we switch to the $LR$ policy from the time slot $t_{b-1}$ and onwards. Let $\pi^{(1)}$ be as in section $\ref{t_b-1ON}$ and $\pi^{(2)}$ be a policy that makes the same decisions as $\pi^{(1)}$ until $t_{b-1}$. From that time slot and onwards, $\pi^{(2)}$ agrees with the $LR$ (figure \ref{steps}, Step 2). For the rest of this section, we will refer to $\pi^{(1)}$ and $\pi^{(2)}$ as $\pi$ and $LR$, respectively. We will now show that for the two policies, $LR$ and $\pi$, we have :
\begin{equation}
{t_{b-1}^{ON}}^{(LR)} \leq {t_{b-1}^{ON}}^{(\pi)} \text{ and}
\label{7}
\end{equation}
\begin{equation}
{R_{b-1}^*}^{(LR)}({t_{b-1}^{ON}}^{(\pi)}) \subseteq {R_{b-1}^*}^{(\pi)}({t_{b-1}^{ON}}^{(\pi)})
\label{8}
\end{equation}
i.e., $LR$ will reach faster than $\pi$ at the time slot $t_{b-1}^{ON}$ and when $\pi$ reaches the milestone time slot (${t_{b-1}^{ON}}^{(\pi)}$), the set of the bottleneck receivers $R_{b-1}^*$ under $LR$ will be a subset of the corresponding set under $\pi$.

Let $t'$ denote the \textit{first} time slot that $D^{\pi}(t') \neq D^{LR}(t')$. We will omit the policy superscript for $t \leq t'$ since the receiver sets that we will examine are the same under both policies. We assume that $t_{b-1} \leq t' < t_{b-1}^{ON}$ \footnote{The case where $t' = t_{b-1}^{ON}$ was covered in section \ref{t_b-1ON}.}. Since $D^{LR}(t') = b-1$, 
\begin{center}
$X_i^{LR}(t'+1) \geq X_i^{\pi}(t'+1) \hspace{1.5 em} \forall i : r_i \in {R_{b-1}^*}(t')$
\end{center}

The strict equality in the equation above holds for any $r_i$ that was disconnected at the time slot $t'$. 
Since $LR$ will favour batch $b-1$, the above equation will hold for the subsequent time slots for every receiver $r_i$ that \textit{is in both sets ${R_{b-1}^*}^{(LR)}(t)$ and ${R_{b-1}^*}^{(\pi)}(t)$} (Figure \ref{state11_1}). Thus:
\begin{center}
$X_i^{LR}(t+1) \geq X_i^{\pi}(t+1)$\\
$\forall i : r_i \in {R_{b-1}^*}^{(\pi)}(t) \cap {R_{b-1}^*}^{(LR)}(t)$, $\hspace{1 em} t \geq t' $
\end{center}

Furthermore, under the $LR$ policy, if a receiver does not belong to the set ${R_{b-1}^*}^{(LR)}(t)$ then it will either belong to the set ${R_b}^{(LR)}(t)$ or $R_{F}^{(LR)}(t)$. Therefore, the equation above will hold for such receivers as long as they also belong in ${R_{b-1}^*}^{(\pi)}(t)$, under $\pi$. Thus,
\begin{equation}
X_i^{LR}(t+1) \geq X_i^{\pi}(t+1), \hspace{1 em} \forall i : r_i \in {R_{b-1}^*}^{(\pi)}(t)
\label{6}
\end{equation}

\begin{figure}
\centering
\includegraphics[scale = 0.28]{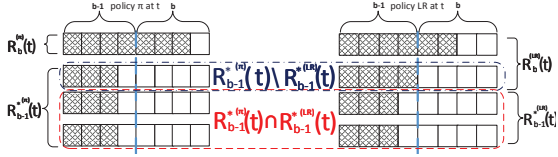}
\caption{System under $\pi$ and $LR$ at $t$ (Section \ref{t_b-1}).}
\label{state11_1}
\end{figure}

\noindent \textbf{Claim:} Since no new receivers can enter either the set ${R_{b-1}^*}^{(\pi)}(t)$ or ${R_{b-1}^*}^{(LR)}(t)$:
\begin{equation}
{R_{b-1}^*}^{(LR)}(t+1) \subseteq {R_{b-1}^*}^{(\pi)}(t+1), \hspace{1.5 em} t \geq t'    
\label{6.1}
\end{equation}
%
%
We will prove this statement using contradiction. Let $r_k$ be a receiver that belongs to the set ${R_{b-1}^*}^{(LR)}(t)$ and does not belong to the set ${R_{b-1}^*}^{(\pi)}(t)$, at some $t \geq t'$.
\begin{itemize}

\item Since $r_k$ does not belong to  ${R_{b-1}^*}^{(\pi)}(t)$, then $r_k$ will either belong to the set ${R_{b}}^{(\pi)}(t)$ or ${R_{F}}^{(\pi)}(t)$. Therefore, $X_k^{(LR)}(t) < X_k^{(\pi)}(t)$.

\item By the definition of $t'$ (and $t_{b-1}$), $X_k^{(LR)}(t') = X_k^{(\pi)}(t')$.
\end{itemize}

Thus, there will exist a $t''$ ($t' \leq t'' < t$) where $X_k^{(LR)}(t'') = X_k^{(\pi)}(t'')$ and $X_k^{(LR)}(t'' + 1) < X_k^{(\pi)}(t'' + 1)$. 
Eq. \ref{6} contradicts the existence of such an $r_k$. Hence, our claim (eq. \ref{6.1}) holds.

\noindent From eq. \ref{6} and \ref{6.1}, we can see that eq. \ref{7} and \ref{8} hold.\\

The next step is to show that if eq. \ref{7} and \ref{8} hold, the completion time of policy LR is not greater than that of $\pi$. Towards that purpose, we will introduce the following lemma.

\begin{lemma}
 Let $\widetilde{\pi}$ and $\pi$ be two policies that make the delay optimal decision at ${t_{b-1}^{ON}}^{(\widetilde{\pi})}$ and ${t_{b-1}^{ON}}^{(\pi)}$, respectively. If both of the statements below are true, then $T^{\widetilde{\pi}} \leq T^{\pi}$.
\begin{enumerate}
\item ${t_{b-1}^{ON}}^{(\widetilde{\pi})} \leq {t_{b-1}^{ON}}^{(\pi)}$
\item ${R_{b-1}^*}^{(\widetilde{\pi})}({t_{b-1}^{ON}}^{(\pi)}) \subseteq {R^*_{b-1}}^{(\pi)}({t_{b-1}^{ON}}^{(\pi)})$ \footnote{Note here that we are comparing the set ${R^*_{b-1}}$ of policy $\widetilde{\pi}$ with the one of policy $\pi$ at time ${t_{b-1}^{ON}}^{(\pi)}$ (which is defined by policy $\pi$).}
\end{enumerate}
\label{lemma1}
\end{lemma}
\begin{proof}[Proof :] 
The intuition for the lemma is that $\widetilde{\pi}$ is in a more advantageous position at time ${t_{b-1}^{ON}}^{(\pi)}$ and therefore completes the transfer earlier. The proof is in line with the arguments of section \ref{t_b-1ON} and can be found in Appendix B.
\end{proof}
Therefore, by using eq. \ref{7}, \ref{8}, lemma \ref{lemma1} and the findings of section \ref{t_b-1ON} , we can see that : 
\begin{center}
$T^{(\pi^{(2)})} \leq T^{(\pi^{(1)})}$
\end{center}
Thus, given any policy $\pi$ the file transfer time is reduced when $\pi$ is modified to become $LR$ from $t_{b-1}$ and onwards.

\subsubsection{\textbf{\underline{Milestone time slot $t_{b-2}^{ON}$}}} \label{t_b-2ON} 

In this subsection we assume the system to be in state $S_{b-2}$. This situation is depicted in figure \ref{state2}, where the set $R_F(t)$ is omitted since it does not affect our analysis. As before, in $S_{b-2}$, the bottleneck receivers are now at the end of batch $b-2$ (set $R_{b-2}^*(t)$). 
Let $\pi^{(2)}$ be as in section \ref{t_b-1} and $\pi^{(3)}$ be a policy that agrees with $\pi^{(2)}$ until $t_{b-2}^{ON}$ and with $LR$ afterwards (refer to figure \ref{steps}, Step 3). We will show that $T^{\pi^{(3)}} \leq T^{\pi^{(2)}}$. Hereinafter, for notational convenience, we will refer to $\pi^{(2)}$ and $\pi^{(3)}$ as $\pi$ and $LR$, respectively. At time $t_{b-2}^{ON}$ all bottleneck receivers will be ON. At that time the decision of policy $\pi$ can be either the same as the decision of $LR$ (i.e. transmit batch $b-2$) or not.

\begin{figure}
\centering
\includegraphics[scale = 0.4]{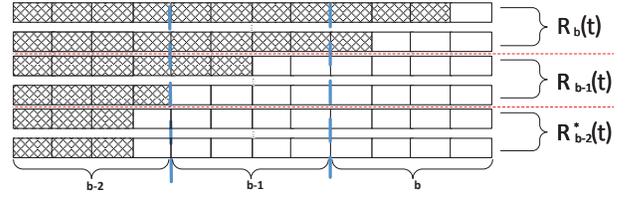}
\caption{State $S_{b-2}$, $K=4$}
\label{state2}
\end{figure}


\vspace{1 em}
\noindent \textit{Case 1) $D^{\pi}(t_{b-2}^{ON}) \neq D^{LR}(t_{b-2}^{ON})$}.

At $t_{b-2}^{ON}+1$, under $LR$, all of the receivers of the set $R_{b-2}^*(t_{b-2}^{ON})$ will be moved to the set $ R_{b-1}(t_{b-2}^{ON} + 1)$. For $t > t_{b-2}^{ON} + 1$, $D^{LR}(t) = b-1$ if at least one of the receivers of the set $R_{b-1}^{(LR)}(t)$ is ON at $t$. 
\begin{itemize}
\item Let $i$ be such that $r_i \in R_{b-2}^*(t_{b-2}^{ON})$.
\item Let $j$ be such that $r_j \in R_{b-1}(t_{b-2}^{ON}) \cup R_{b}(t_{b-2}^{ON})$.
\end{itemize}
The $LR$ policy will transmit a packet to ALL receivers $r_i$ whereas $\pi$ will sent a packet to some receivers $r_j$ (depending on the decision) that are ON. Thus,
\begin{center}
$X_i^{LR}(t_{b-2}^{ON}+1) > X_i^{\pi}(t_{b-2}^{ON}+1)$,\hspace{2 em}
$\forall i$ as above
\end{center}

Figure \ref{state22} depicts the system under policy $\pi$ at time $t_{b-2}^{ON} + 1$ (left) and under policy $LR$ (right). In the middle we can see the connectivities of the receivers at time $t_{b-2}^{ON}$. Policy $\pi$ might transmit batch $b-1$ (red) or $b$ (green) whereas the $LR$ policy will transmit batch $b-2$ (blue). We keep the sets of the time slot $t_{b-2}^{ON}$ in order to illustrate our point. As we can see from the figure, all $r_i$ will be moved to set $R_{b-1}$ at $t_{b-2}^{ON} + 1$ under policy $LR$. Under policy $\pi$, some of the receivers $r_j$ might change sets but the set $R_{b-2}^*(t_{b-2}^{ON})$ will not be affected.

\begin{figure}
\centering
\includegraphics[scale = 0.3]{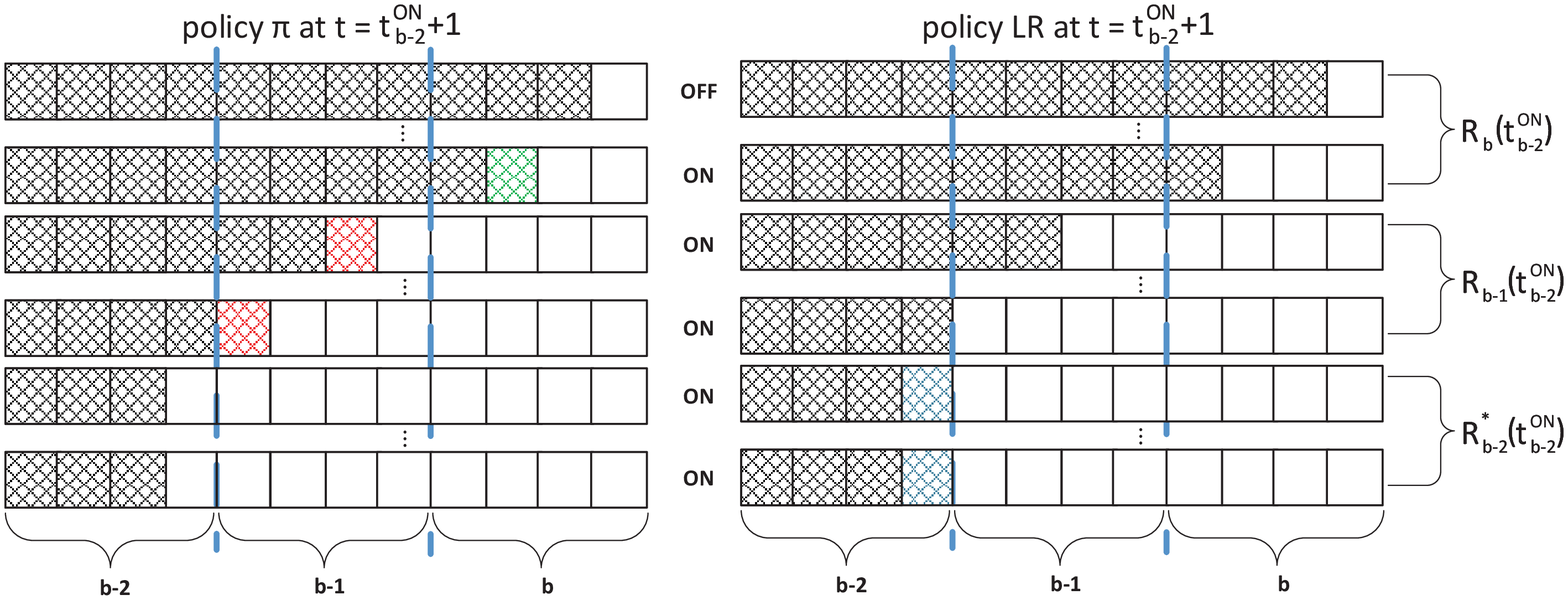}
\caption{System at time $t_{b-2}^{ON} + 1$ with $\pi$ (left) and $LR$ (right).}
\label{state22}
\end{figure}

Under $LR$, each $r_i$ will receive a packet at each time slot that $r_i$ is ON, until $r_i$ successfully decodes batch $b-1$ (i.e. enters set $R_b(t)$). Under $\pi$, a receiver $r_i$ \textit{might} receive a packet of batch $b-1$. Hence,
\begin{equation}
\begin{split}
X_{i}^{LR}(t+1) > X_{i}^{\pi}(t+1), \hspace{1.5 em} t \geq t_{b-2}^{ON},\\ \forall i : r_{i} \in R_{b-2}^*(t_{b-2}^{ON})\cap R_{b-1}^{(LR)}(t) \hspace{1em}
\label{9}
\end{split}
\end{equation}
{(i.e. for all $r_i$'s that have not yet entered batch $b$, under $LR$.)}\\
Note : In eq. \ref{9} we used $t+1$ in order to show that this hold up to and including the time slot that $r_i$  \textit{enters} the set $R_b(t)$ under $LR$.

From eq. \ref{9} we can observe that, for each $r_i$, the required time to enter the set $R_{b-1}^*$ (i.e. be one packet short of decoding batch $b-1$) will be smaller under $LR$ than under $\pi$ (since each receiver will have received more packet under $LR$).

By applying lemma \ref{lemmaA} with the following parameters,
\begin{itemize}
\item $t_s = t_{b-2}^{ON} + 1$.
\item  $R_{\beta(t_s)}^{LR} = R_{b-1}^{LR}(t_{b-2}^{ON}+1)$.
\item $R_i^{LR} = R_{b-2}^*(t_{b-2}^{ON})$.
\item $R_j^{LR} = R_{b-1}(t_{b-2}^{ON})$.
\end{itemize}
we can show that, for $i$ as above and $j : r_j \in R_{b-1}(t_{b-2}^{ON}) \cap R_{b-1}^{(LR)}(t)$:
\begin{equation}
\min\limits_{i} X_{i}^{LR}(t+1) \leq \min\limits_{j} X_{j}^{LR}(t+1) \hspace{2 em} t \geq t_{b-2}^{ON}
\label{10}
\end{equation}
\noindent Hence, based on eq. \ref{9} and \ref{10}, we can see that, for $t \geq t_{b-2}^{ON}$:
\begin{center}
$\min\limits_{i} X_{i}^{\pi}(t+1) < \min\limits_{i} X_{i}^{LR}(t+1) \leq \min\limits_{j} X_{j}^{LR}(t+1)$
\end{center}
Therefore, at time $t_{b-1}^{(\pi)}$, under $LR$, no receiver will be in the set $R_{b-1}(t_{b-1}^{(\pi)})$ (${R_{b-1}^*}^{(LR)}({t_{b-1}}^{(\pi)}) = \emptyset$)\footnote{This holds due to the strict inequality of eq. \ref{9}}. Thus,
\begin{center}
$t_{b-1}^{(LR)} < t_{b-1}^{(\pi)}$ and\\
${R_{b-1}^*}^{(LR)}({t_{b-1}}^{(\pi)}) \subseteq {R_{b-1}^*}^{(\pi)}({t_{b-1}}^{(\pi)})$
\end{center}

\vspace{1 em}
\noindent \textit{Case 2) $D^{\pi}(t_{b-2}^{ON}) = D^{LR}(t_{b-2}^{ON}) = b-2$}

In this case\footnote{This part also covers the case where $t_{b-2}^{ON}$ is not a conflict slot.}, all of the receivers of the set $R_{b-2}^*(t_{b-2}^{ON})$ will be moved to the set $R_{b-1}(t_{b-2}^{ON} + 1)$, under both policies. 
\begin{figure}
\centering
\includegraphics[scale = 0.4]{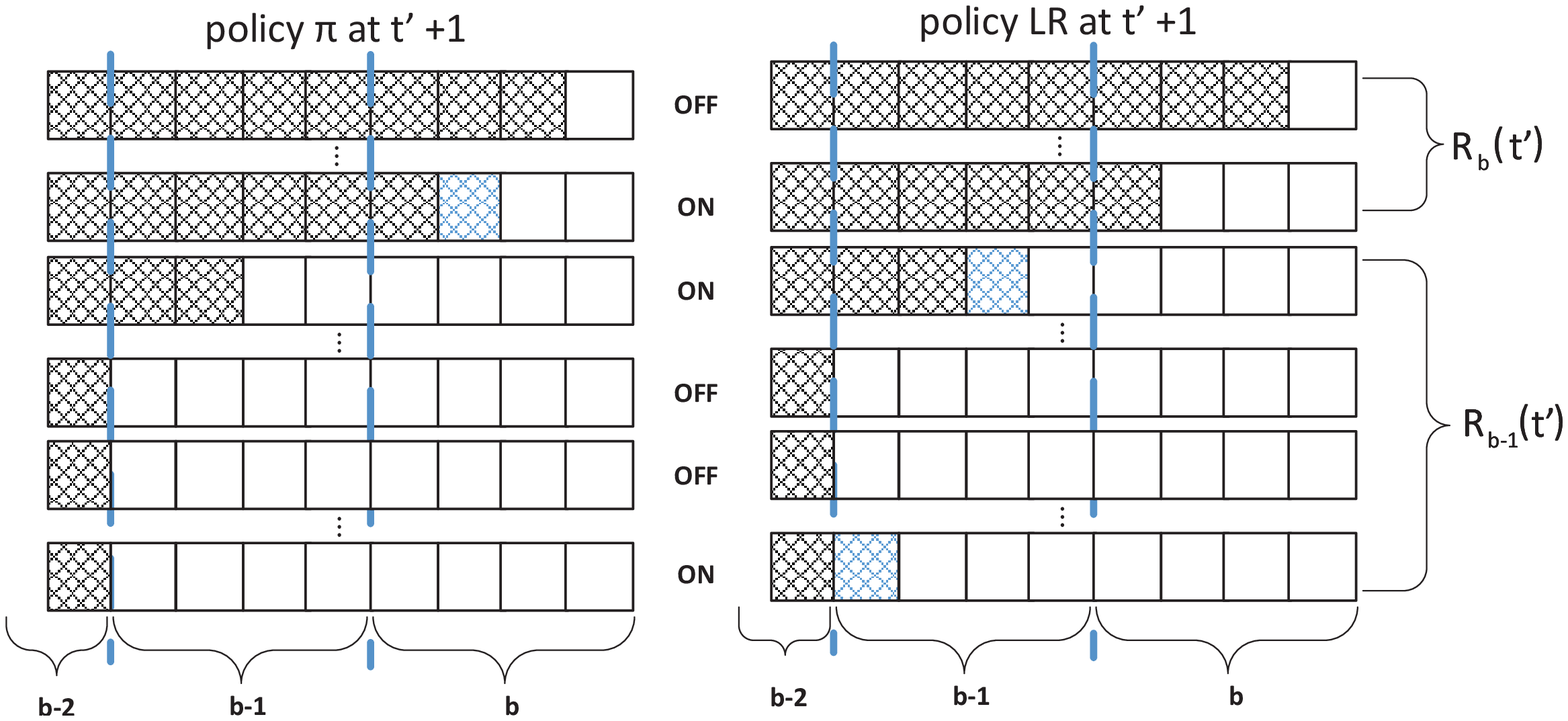}
\caption{System at time $t'$ - section \ref{t_b-2ON}, Case 2.}
\label{state22-case2}
\end{figure}
For $t \geq t_{b-2}^{ON} +1$ the system under both policies will be the same until $\pi$ makes a different decision than $LR$. Let $t'$ be the first time slot (after $t_{b-2}^{ON}$) that $D^{\pi}(t') \neq D^{LR}(t')$ (figure \ref{state22-case2}). Then, \begin{center}
$X_{i}^{LR}(t'+1) \geq X_{i}^{\pi}(t'+1)$ \footnote{The strict inequality is for all $r_i$ that we ON at $t'$ and the equality is for the rest of the $r_i$'s.}, $r_{i} \in R_{b-1}(t')$\footnote{We omitted the policy superscript because the sets $R_{b-1}^{(\pi)}(t')$ and $R_{b-1}^{(LR)}(t')$ are identical.}
\end{center} 
This equation will hold for the subsequent time slots, for the receivers that are in the batch $b-1$ under both policies. Similarly with the derivation of eq. \ref{6}, we can show that:
\begin{equation}
X_{i}^{LR}(t+1) \geq X_{i}^{\pi}(t+1), \hspace{2 em} t \geq t', r_{i} \in R_{b-1}^{(\pi)}(t)
\label{2.0}
\end{equation}
From the equations above\footnote{and from the fact that no receiver can enter either set $R_{b-1}^{(\pi)}(t)$ or $R_{b-1}^{(LR)}(t)$ for $t \geq t_{b-2}^{ON}$}, we can see that since each receiver of the set $R_{b-1}^{(\pi)}(t)$ will have received more (or an equal number of) packets under $LR$ than under $\pi$ from $t_{b-2}^{ON}$ and onwards then,
\begin{equation}
t_{b-1}^{(LR)} \leq t_{b-1}^{(\pi)}
\label{2.1}
\end{equation} 
\begin{equation}
{R_{b-1}^*}^{(LR)}(t_{b-1}^{(\pi)}) \subseteq {R_{b-1}^*}^{(\pi)}(t_{b-1}^{(\pi)})
\label{2.2}
\end{equation}

Hence we can see that in both cases, eq. \ref{2.1} and \ref{2.2} hold. As in section \ref{t_b-1}, the next step is to develop a lemma in order to compare the completion time of two policies that take action prior to the time slot $t_{b-1}$.

\begin{lemma}
 Let $\widetilde{\pi}$ and $\pi$ be two policies that make the delay optimal decisions starting from ${t_{b-1}}^{(\widetilde{\pi})}$ and ${t_{b-1}}^{({\pi})}$, respectively. If both of the statements below are true, then $T^{\widetilde{\pi}} \leq T^{\pi}$.
\begin{enumerate}
\item ${t_{b-1}}^{(\widetilde{\pi})} \leq {t_{b-1}}^{(\pi)}$
\item ${R_{b-1}^{*}}^{(\widetilde{\pi})}({t_{b-1}}^{(\pi)}) \subseteq {R_{b-1}^*}^{(\pi)}({t_{b-1}}^{(\pi)})$
\end{enumerate}
\label{lemma2}
\end{lemma}
\begin{proof}[Proof :]
The proof of this Lemma is similar to the proof of lemma \ref{lemma1} and can be found in Appendix C.
\end{proof}
Using lemma \ref{lemma2}, we can see that :
\begin{center}
$T^{\pi^{(3)}} \leq T^{\pi^{(2)}}$
\end{center}

%

\subsubsection{\textbf{\underline{Milestone time slot $t_{b-2}$}}} \label{t_b-2}

At $t_{b-2}$, some of the bottleneck receivers ($R_{b-2}^{*}(t)$) may be OFF. Clearly, $t_{b-2}$ $\leq t_{b-2}^{ON}$.  Let $\pi^{(3)}$ be as in section \ref{t_b-2ON} and $\pi^{(4)}$ be a policy that agrees with $\pi^{(3)}$ until $t_{b-2}$ and with $LR$ at and after that time slot (figure \ref{steps}, Step 4). We will show that $T^{\pi^{(4)}} \leq T^{\pi^{(3)}}$. For the rest of this section, we will refer to $\pi^{(3)}$ and $\pi^{(4)}$ as $\pi$ and $LR$, respectively.
We denote with $t'$ the first time slot ($t_{b-2} \leq t' < \min\{{t_{b-2}^{ON}}^{(LR)},{t_{b-2}^{ON}}^{(\pi)}\}$) where the decisions of $LR$ and $\pi$ are different (i.e. $D^{\pi}(t') \neq D^{LR}(t')$). Depending on the value of $D^{LR}(t')$ we will distinguish two cases.\\

\noindent \textit{Case 1) $D^{LR}(t') = b-2$}

In this case, the system under both policies is the same up to the time slot $t'$ and the $LR$ policy will transmit batch $b-2$ at $t'$:
\begin{equation*}
X_i^{LR}(t'+1) \geq X_i^{\pi}(t'+1)\footnote{As in section \ref{t_b-1}, the equality is for any $r_i$ that was disconnected at $t'$.} \hspace{1.5 em} \forall i : r_i \in {R_{b-2}^*}(t')
\end{equation*}

It is easy to see that, since $LR$ will transmit batch $b-2$ \textit{at each} time slot that at least one of the $r_i$'s is ON (up to and including time slot ${t_{b-2}^{ON}}^{(LR)}$), the above equation will hold for the subsequent time slots. That is : 
\begin{equation}
X_i^{LR}(t+1) \geq X_i^{\pi}(t+1) \hspace{1.5 em} t \geq t', \forall i : r_i \in {R_{b-2}^*}^{(\pi)}(t) \footnote{The reason that this equation holds for $r_i \in {R_{b-2}^*}^{(\pi)}(t)$ is analogous to the one analysed in section \ref{t_b-1} eq. \ref{6}.}
\label{11}
\end{equation}

At each $t$, every $r_i$ that is ON will \textit{surely} receive a packet with $LR$ (and thus move to set $R_{b-1}^{(LR)}(t+1)$), whereas it \textit{might} receive a packet with $\pi$. Similarly to section \ref{t_b-1}, we can show that :

\begin{equation}
{R_{b-2}^*}^{(LR)}(t+1) \subseteq {R_{b-2}^*}^{(\pi)}(t+1), \hspace{1.5 em} t > t'
\label{12}
\end{equation}
\begin{equation}
{t_{b-2}^{ON}}^{(LR)} \leq {t_{b-2}^{ON}}^{(\pi)}
\label{13}
\end{equation}
and from eq \ref{12} :
\begin{equation}
{R_{b-2}^*}^{(LR)}({t_{b-2}^{ON}}^{(\pi)}) \subseteq {R_{b-2}^*}^{(\pi)}({t_{b-2}^{ON}}^{(\pi)})
\label{14}
\end{equation}

As in section \ref{t_b-2ON}, a lemma needs to be introduced in order to compare the completion times of two policies that act prior to $t_{b-2}^{ON}$.

\begin{lemma}
Let $\widetilde{\pi}$ and $\pi$ be two policies that make the delay optimal decisions starting from ${t_{b-2}^{ON}}^{(\widetilde{\pi})}$ and ${t_{b-2}^{ON}}^{({\pi})}$, respectively. If both of the statements below are true, then $T^{\widetilde{\pi}} \leq T^{\pi}$.
\begin{enumerate}
\item ${t_{b-2}^{ON}}^{(\widetilde{\pi})} \leq {t_{b-2}^{ON}}^{(\pi)}$ and
\item ${R_{b-2}^{*}}^{(\widetilde{\pi})}({t_{b-2}^{ON}}^{(\pi)}) \subseteq {R_{b-2}^*}^{(\pi)}({t_{b-2}^{ON}}^{(\pi)})$
\end{enumerate}
\label{lemma3}
\end{lemma}
\begin{proof}[Proof of Lemma \ref{lemma3} :] $ $\newline
The proof of this Lemma is inline with the arguments of section \ref{t_b-2ON} and lemma \ref{lemma2}.

\end{proof}

Therefore, by using eq. \ref{13} and \ref{14} and lemma \ref{lemma3} we can show that, for case 1, $T^{\pi^{(4)}} \leq T^{\pi^{(3)}}$.\\

\noindent \textit{Case 2\footnote{In this case, all of the receivers of the set $R^*_{b-2}(t')$ are OFF}) $D^{LR}(t') = b-1$}.

By the definition of $t'$, we know that $D^{\pi}(t') = b$. 
Let $t''$ be the first time slot after $t'$ where $D^{LR}(t'') \neq D^{\pi}(t'')$ \textit{and} $D^{LR}(t'') = b-2$ ($t''$ might not exist). Regardless of when (and if) $t''$ occurs, we know that the receivers of the set ${R_{b-2}^*}(t)$ (for $t_{b-2} < t \leq t'' $) will not be affected by $t'$. For each of the following cases we can use the analysis of previous sections with straightforward modifications. Therefore, in order to save space, we outline the proof by referring the reader to the corresponding prior section.
\begin{itemize}
\item \textit{Subcase 2.1) $t'' < \min\{{t_{b-2}^{ON}}^{(LR)},{t_{b-2}^{ON}}^{(\pi)}\}$}

Section \ref{t_b-2} - Case 1

\item \textit {Subcase 2.2) $t'' = {t_{b-2}^{ON}}$ or $t''$ does not exist}

In both of these cases, we can see that 
\begin{center}
${t_{b-2}^{ON}}^{(LR)} = {t_{b-2}^{ON}}^{(\pi)}$ and \\${R_{b-2}^{*}}^{(LR)}({t_{b-2}^{ON}}^{(\pi)}) = {R_{b-2}^*}^{(\pi)}({t_{b-2}^{ON}}^{(\pi)})$
\end{center}
If $t'' = {t_{b-2}^{ON}}$, then the analysis is analogous to section \ref{t_b-2ON}-1. If $t''$ does not exist, the analysis is similar to section \ref{t_b-2ON}-2.
\end{itemize}
Given the previous results of section \ref{t_b-2}, the file transfer completion time of any policy $\pi$ can be reduced when $\pi$ is modified to agree with $LR$ from $t_{b-2}$ and onwards. Hence, \\
\centerline{$T^{\pi^{(4)}} \leq T^{\pi^{(3)}}$.}

\vspace{1 em}
By using \textit{backwards induction}, the policy improvement can be performed for the rest of the milestone time slots $t_{b-3}^{ON}, t_{b-3}, ..., t_{1}$ (analogous to sections \ref{t_b-2ON} and \ref{t_b-2}). Therefore, at each step, we can construct a policy $\pi^{(n)}$ (that agrees with $LR$ from $t^{(n)}$ and onwards) such that 
\begin{center}
$T^{\pi^{(n)}} \leq T^{\pi^{(n-1)}}$, $n=1, ..., 2b$\\ where
$\pi^{(0)}$ is the initial policy and $\{t_{b-1}^{ON}, t_{b-1}, ..., t_{b-(b-1)}\} = \{t^{(1)}, t^{(2)}, ..., t^{(2b)}\}$
\end{center}

Moreover, inductively, we can generalize the Lemmas \ref{lemma1}, \ref{lemma2}, \ref{lemma3} in the following form:

\begin{lemma}
Let $\{t^{(1)}, t^{(2)}, ..., t^{(2b)}\} = \{t_{b-1}^{ON}, t_{b-1}, ..., t_{1}\}$.

Let $\widetilde{\pi}$ and $\pi$ be two policies that make the delay optimal decisions starting from ${t^{(n)}}^{(\widetilde{\pi})}$ and ${t^{(n)}}^{({\pi})}$, respectively. If both of the statements below are true, then $T^{\widetilde{\pi}} \leq T^{\pi}$.
\begin{enumerate}
\item ${t^{(n)}}^{(\widetilde{\pi})} \leq {t^{(n)}}^{(\pi)}$ and
\item ${R_{b-2}^{*}}^{(\widetilde{\pi})}({t^{(n)}}^{(\pi)}) \subseteq {R_{b-2}^*}^{(\pi)}({t^{(n)}}^{(\pi)})$
\end{enumerate}
\label{lemma5}
\end{lemma}

The proof of theorem \ref{thm} is concluded by examining the milestone time slot $t^{(0)}$. We define $t^{(0)}$ as the first time slot where at least one of the receivers receives the first $K$ packets. It can be shown that if $\pi^{(2b+1)}$ agrees with $LR$ from $t^{(0)}$ and onwards and $\pi^{(2b)}$ agrees with $LR$ from $t^{(1)} = t_1$ and onwards then $T^{\pi^{(2b+1)}} \leq T^{\pi^{(2b)}}$.

At $t=0$, all of the receivers will have empty queues. In the interval $[t=0,t^{(0)}]$ all of the policies will make the same decision since all of the receivers are expecting packets of the first batch. At and after $t^{(0)}$, the $LR$ policy ($\pi^{(2b+1)}$) will send packets to the connected receiver with the least number of received packets \textit{at all times} and $\pi$ ($\pi^{(2b)}$) transmits an encoded packet of any of the available batches. From the definition of the policies, we can see that : 
\begin{equation*}
X_i^{(LR)}(t^{(0)}+1) \geq X_i^{(\pi)}(t^{(0)}+1), \hspace{1.5 em} i : r_i \in {R_1}(t^{(0)})   
\end{equation*}
As before, the same equation will apply for any $t>t^{(0)}$ for the receivers that are in the set $R_{1}(t)$ under both policies :
\begin{equation*}
X_{i}^{(LR)}(t+1) \geq X_{i}^{(\pi)}(t+1), t \geq t^{(0)}, r_{i} \in {R_1}^{(LR)}(t) \cap {R_1}^{(\pi)}(t)
\end{equation*}
Any $r_i$ that belongs to ${R_1}^{(\pi)}(t)$ and does not belong to ${R_1}^{(LR)}(t)$ satisfies the equation above. Thus : 
\begin{equation}
X_{i}^{(LR)}(t+1) \geq X_{i}^{(\pi)}(t+1), t \geq t^{(0)}, r_{i} \in {R_1}^{(\pi)}(t)
\label{15}
\end{equation}
Therefore it is easy to see that : 
\begin{equation}
t_1^{(LR)} \leq t_1^{(\pi)}
\label{16}
\end{equation}
and 
\begin{equation}
{R_1}^{(LR)}(t_1^{(\pi)}) \subseteq {R_1}^{(\pi)}(t_1^{(\pi)})
\label{17}
\end{equation}
The proof of theorem \ref{thm} is concluded by using lemma \ref{lemma5} for $n=2b$.
%
\section{Expected File Transfer Completion Time}\label{approx}
In \cite{skevakis2016decoding} we provided a closed form approximation for the expected file transfer completion time under the policy $LR$. Furthermore, we derived a formula for the minimum coding window size so that the expected file transfer completion time, under $LR$, is upper bounded by a user defined delay constraint. We will briefly summarize our main results of \cite{skevakis2016decoding} and we will then provide new extensions and further approximations.

Let $X$ be a Gaussian random variable with mean $\mu = \frac{K}{p}$ and standard deviation $\sigma = \sqrt{\frac{K(1-p)}{p^2}}$, where $K$ is the coding window size and $p$ is the probability that a receiver is connected to the base station. Then, $X$ will accurately represent the file transfer completion time\footnote{for large $K$ (as $K$ increases the accuracy of our approximation increases) and moderate $p$ (in our experiments in section \ref{exps} we consider $p \in [0.1:0.9]$).} of a \textit{single} receiver (\cite{skevakis2016decoding}). The expected file transfer completion time, 
when the coding window is the entire file ($K=F$), will be :
\begin{equation}
\mathbb{E}[T_K] = \mathbb{E}[\max\limits_{1 \leq i \leq N}X_i] \approx \int_{0}^\infty (1 - (F_X(z))^N) dz,
\label{ETKK}
\end{equation}
where $F_X(z)$ represents the cdf of the Gaussian random variable $X$ and $N$ is the total number of receivers.
We showed (we refer the reader to \cite{skevakis2016decoding} for more details) that the completion time of a file of $F$ packets when using a coding window size of $K$ ($\mathbb{E}[T^F_K]$), under the $LR$ policy, can be approximated by:
\begin{equation}
\mathbb{E}[T^F_K] \approx b*\mathbb{E}[T_K] 
\label{ETFK}
\end{equation}
where $b = \frac{F}{K}$ is the total number of batches. Based on the \textit{$3$-sigma rule} of the Gaussian distribution, we showed that 
\begin{equation}
\mathbb{E}[T_K^F] \approx  b\mu + b\widetilde{n}\sigma - b\sigma A(N)
\label{ETFK1}
\end{equation}
where $\widetilde{n} \triangleq \widetilde{n}(N) = \min\{n: (erf^N(\frac{n}{\sqrt{2}}))\geq 0.99\}$, $K > \widetilde{n}^2(1-p)$ and $A(N) = \int_{- \widetilde{n}}^{\widetilde{n}} (\int_{- \widetilde{n}}^z \frac{1}{\sqrt{2\pi}}e^{-\frac{t^2}{2}} dt)^N dz$.\footnote{The constraint $K > \widetilde{n}^2(1-p)$ is due to the assumption that $\mu - \widetilde{n}\sigma > 0$. For representative values of $\widetilde{n}, K$ we refer the reader to \cite{skevakis2016decoding}.}

From our experiments, we observed that \textit{near optimal} completion time can be achieved with a coding window size $K \ll F$. Therefore, we derived a formula to determine the minimum achievable coding window size that provides an acceptable increase in the completion time (w.r.t. the optimal delay for $K=F$ \cite{joshi2013round}). We showed that, given a user defined delay constraint $\epsilon$ (expressed as a percentage of increase of the minimum achievable completion time $\mathbb{E}[T_{opt}] = \mathbb{E}[T^F_F]$) the minimum coding window size will satisfy the following formula :
\begin{equation}
\frac{\mathbb{E}[T_K^F] - \mathbb{E}[T_{opt}]}{\mathbb{E}[T_{opt}]} = \frac{\sqrt{1 - p}(\widetilde{n} - A(N))}{\sqrt{F} + \sqrt{1 - p}(\widetilde{n} - A(N))}(\sqrt{\frac{F}{K}} - 1) \leq \epsilon
\label{epsilon}
\end{equation}

The accuracy of eq. \ref{ETFK} and \ref{epsilon} is satisfactory, as we showed in \cite{skevakis2016decoding}. However, motivated by the fact that the computation of $\widetilde{n}$ and $A(N)$ is not straightforward, we now provide computational enhancements for $A(N)$ and eq. \ref{ETFK1} and \ref{epsilon}.

\subsection{Approximation for $A(N)$}
 
By the definition of  $\widetilde{n}$ we know that ${erf}^N(\frac{\widetilde{n}}{\sqrt{2}}) \geq 0.99$. Thus, $\frac{1}{\sqrt{2\pi}}\int_{-\infty}^{-\widetilde{n}}e^{-x^2/2}dx \approx 0 $.
Therefore, 
\begin{center}
$A(N) = \int_{- \widetilde{n}}^{\widetilde{n}} (\int_{- \widetilde{n}}^z \frac{1}{\sqrt{2\pi}}e^{-\frac{t^2}{2}} dt)^N dz = $ \\$= \int_{- \widetilde{n}}^{\widetilde{n}} (\int_{-\infty}^z \frac{1}{\sqrt{2\pi}}e^{-\frac{t^2}{2}} dt - \int_{-\infty}^{\widetilde{n}} \frac{1}{\sqrt{2\pi}}e^{-\frac{t^2}{2}} dt)^N dz \approx $\\$\int_{- \widetilde{n}}^{\widetilde{n}} (\int_{-\infty}^z \frac{1}{\sqrt{2\pi}}e^{-\frac{t^2}{2}} dt)^N = \int_{- \widetilde{n}}^{\widetilde{n}} \Phi(x)^N dx$,
\end{center}
where $\Phi(x)$ is the cdf of the standard normal $\mathcal{N}(0,1)$. Moreover, we know that, for the Q-function, $\Phi(x) = 1 - Q(x)$ and $Q(-x) = 1 - Q(x)$. Thus,
\begin{center}
$A(N) = \int_{- \widetilde{n}}^{\widetilde{n}} \Phi(x)^N dx = \int_{0}^{\widetilde{n}} (Q(x)^N + \Phi(x)^N) dx$
\end{center}
The Q-function is monotonically decreasing and as $N$ increases, $Q(x)^N$ rapidly decreases. Therefore, we can assume that $\int_{0}^{\widetilde{n}} Q(x)^Ndx \approx 0$, with improved accuracy as $N$ increases\footnote{For $\widetilde{n} = 5, \int_{0}^{\widetilde{n}} Q(x)^Ndx = 0.1168, 0.0164, 0.0029$ for $N = 2, 4, 6$, respectively.}. Hence, 
\begin{equation}
A(N) \approx \int_{0}^{\widetilde{n}} \Phi(x)^Ndx
\label{Alpha}
\end{equation}

\subsection{Approximations for $\mathbb{E}[T_K^F]$ and $K$ in eq. \ref{ETFK1} and \ref{epsilon}}

As we argued for eq. \ref{ETKK}, the file transfer completion time, when $K=F$, is the expected value of the maximum of $N$ Gaussian random variables (denoted with $X_i$) with mean $\mu = \frac{K}{p}$ and standard deviation $\sigma = \sqrt{\frac{K(1-p)}{p^2}}$. Therefore,
\begin{equation*}
\mathbb{E}[T_K] =  \mathbb{E}[\max\limits_{1 \leq i \leq N}X_i] = \mu + \sigma \mathbb{E}[\max\limits_{1 \leq i \leq N}Z_i]
\end{equation*}
where $Z \backsim \mathcal{N}(0,1)$. In \cite{chen1999accurate}, the authors accurately approximated the expected value of the greatest order statistics for Gaussian r.v.'s by the expression $\Phi^{-1}(0.5264^{\frac{1}{N}})$, where $\Phi^{-1}$ is the inverse of the Gaussian cdf and $N$ is the sample size. Even though $\Phi^{-1}$ has no closed form representation, many polynomial approximations exist (\cite{wichura1988}, \cite{acklam2003}) and it is also built in many commercial mathematical software packages (e.g. MATLAB, Mathematica). We will use the findings of \cite{chen1999accurate} in order to derive simpler formulas for eq. \ref{ETFK1} and \ref{epsilon}. Based on the above, eq. \ref{ETKK} can be rewritten as:
\begin{equation*}
\mathbb{E}[T_K] \approx \mu + \sigma*\Phi^{-1}(0.5264^{1/N})
\end{equation*}
and by using eq. \ref{ETFK}
\begin{equation}
\mathbb{E}[T^F_K] \approx b\mu + b\sigma*\Phi^{-1}(0.5264^{1/N})
\label{ETFKappr}
\end{equation}
Similarly with the derivation of the above equations,
\begin{center}
$\mathbb{E}[T_{opt}] = \mathbb{E}[T^F_F] = \mu_F + \sigma_F*\Phi^{-1}(0.5264^{1/N}) $,
\end{center}
where $\mu_F = F/p$ and $\sigma_F = \sqrt{\frac{F(1-p)}{p^2}}$. Thus, eq. \ref{epsilon} will be transformed to
\begin{equation}
\frac{(\sqrt{\frac{F}{K}}-1)(\sqrt{1-p})B(N)}{\sqrt{F} + (\sqrt{1-p})B(N)} \leq \epsilon,
\label{eBNF}
\end{equation}
where $B(N) = \Phi^{-1}(0.5264^{1/N})$. Furthermore, in the denominator, the term $\sqrt{F}$  dominates the term $\sqrt(1-p)B(N)$. Therefore, we can further simplify eq. \ref{eBNF} as 
\begin{equation}
\frac{(\sqrt{\frac{F}{K}}-1)(\sqrt{1-p})B(N)}{\sqrt{F}} \leq \epsilon,
\label{epsilonappr}
\end{equation}
As we can notice, eq. \ref{ETFKappr} and \ref{epsilonappr} are not only simpler than \ref{ETFK1} and \ref{epsilon}, respectively, but they also do not contain the variable $\widetilde{n}$ and are thus applicable for all $K$'s. Therefore, given $F$, $N$ and $\epsilon$ the minimum value of $K$ can be readily computed.

\section{Experimental Results}\label{exps}
We performed our experiments\footnote{The experiments were performed with a custom built simulator using Java.} under various simulation conditions. Under such conditions we included low, medium and high system load (for $N \leq 10, 50$ and $100$, respectively), different file sizes ($400 \leq F \leq 10K$) and receiver connectivity probabilities in the range $0.1 \leq p \leq 0.9$. Each experiment was repeated multiple times and the averages were calculated for each value of $K$ adjusted so that $F/K$ is integer valued. In the limited available space we will try to give a comprehensive outlook of our simulation based evaluations.
\subsection{Comparison of LR with other policies \protect\footnote{$95\%$ confidence intervals were calculated but since they turned out to be very narrow are omitted.}}

We compared the performance of the $LR$ policy with that of 2 other policies 
A description of those policies follows:

\textit{Random Selection (RS): }This policy is based on randomly selecting a batch. Each batch $i$ is selected with probability $\frac{Ni}{Nc}$, where $Ni$ is the number of connected receivers with batch ID $i$ and $Nc$ is the total number of connected receivers.
%

\textit{Maximum Gain (MG): }This is a "greedy" policy. MG maximizes the \textit{instantaneous} throughout by selecting, at each time slot, the batch that will be beneficial to the largest number of the connected receivers. 

The upper part of figures \ref{fig1} and \ref{fig2} show the file transfer completion time (the \textit{maximum} completion time among the receivers, averaged out of 200 independent replications), normalized by $F/p$ (the average completion time of a receiver in the ideal case, where $K=F$). We restrict the coding window to values less than $F/3$ since the completion time is almost the same for the rest of the values, under all policies. This is expected, since as we increase $K$ the number of conflict slots decreases. Thus, the policies act on less time slots and the effect, of each policy, decreases. We can see that the $LR$ policy largely outperforms the other 2 policies, especially when the coding window size is small. The LR policy achieves $82\% - 91\%$ lower file transfer completion time than the $RS$ and the $MG$, for $K < 180$ $(6\%$ of $F$, figure \ref{fig1}). As the system load and the file size is increased, those percentages increase to $91\% - 97\%$ for the $RS$ and $94\% - 97\%$ for the $MG$ policy, for $K < 600$ $(6\%$ of $F$, figure \ref{fig2}). This behaviour is verified by the rest of our experiments; the improvement in the file transfer completion time, under the $LR$ policy, increases as we increase the values of $N$, $F$ or $p$.

The middle part of figures \ref{fig1} and \ref{fig2} show the average completion time of a receiver (the \textit{average} completion time among the receivers, averaged out of 200 independent replications), normalized by $F/p$. It is interesting to notice that the file transfer completion time, under the $LR$, is almost the same as the average completion time of a receiver. The $LR$ policy, by favouring the receivers with the least number of received packets, manages to reduce the differences in the completion time of the receivers. As a result, the maximum and the average completion time of the receivers is almost the same. The other 2 policies do not exhibit this behaviour; the average completion time is significantly less than the file transfer completion time (maximum completion time). Thus, there is a wide spread in the completion time among the receivers which is verified by the lower parts of figures \ref{fig1} and \ref{fig2}. The lower parts of these figures show the variance of the normalized completion time among the receivers, averaged out of 200 replication. Therefore, the $LR$ policy also enhances the "fairness" of the system; all of the receivers receive the entire file in almost the same time.

\begin{figure}
\centering
\includegraphics[scale = 0.5, trim = 1.5cm 0cm 0cm 0cm, clip]{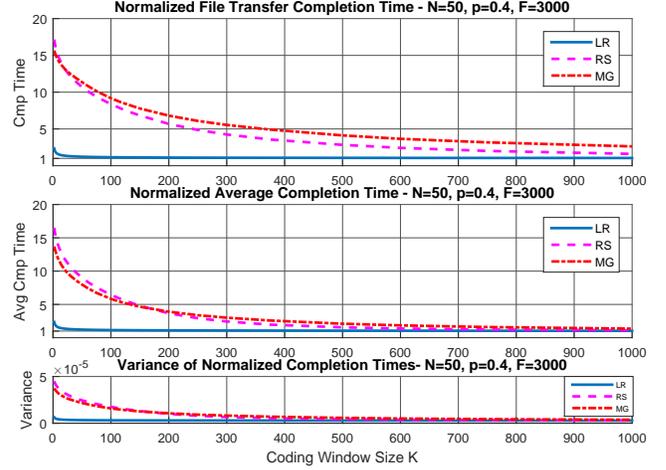}
\caption{File transfer completion time and average completion time of a receiver, under 3 policies.}
\label{fig1}
\end{figure}

\begin{figure}
\centering
\includegraphics[scale = 0.5, trim = 1.5cm 0cm 0cm 0cm, clip]{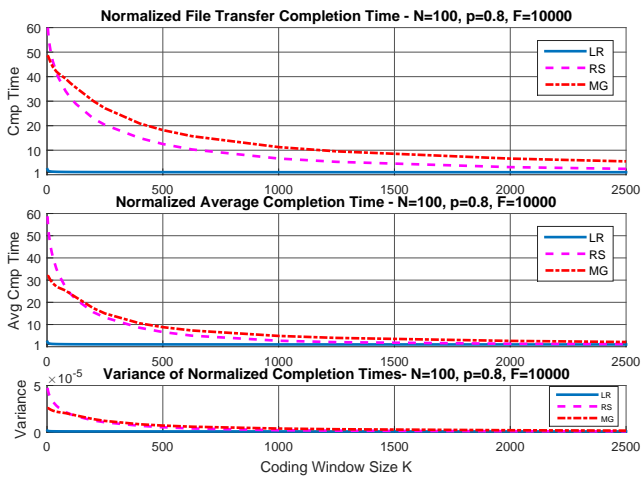}
\caption{File transfer completion time and average completion time of a receiver, under 3 policies.}
\label{fig2}
\end{figure}

Figures \ref{thr1} and \ref{thr2} depict the normalized throughput of the 3 policies. We define the throughput of the system as the average number of received packets per time slot (a transmission of a single packet can result in a maximum of $N$ received packets; all receivers are ON and have the same batch ID). In the ideal case, when $K=F$, no conflict slots will occur and the average throughput will be $Np$. As soon as (at least) one receiver receives the entire file, this value will decrease. In order to perform a fair comparison of the 3 policies, we transmitted a large file ($>10K$ packets) and calculated the long term throughput until the minimum completion time among the receivers. We normalized the calculated throughput by the throughput of the ideal case, $Np$. As we can see from figures \ref{thr1} and \ref{thr2}, the $LR$ policy exploits the broadcast nature of the channel more efficient than the other 2 policies. A normalized throughput of $0.9$ is reached with a coding window size of $200$ (figure \ref{thr1}) and $125$ (figure \ref{thr2}) when $N = 50$, $p = 0.4$ and $N = 100$, $p = 0.8$, respectively. This is of major importance since we can see that the $LR$ policy can achieve a throughput close to the optimal one with relatively small values of $K$. The other 2 policies need a coding window size of approximately 10 times larger in order to achieve a normalized throughput close to $0.75$. The $LR$ policy, by favouring the receivers with the least number of received packets, manages to balance the queues of the receivers. Therefore, on the long run, each transmitted packet is beneficial to more receivers ($0.9$ throughput means that, on average, $90\%$ of the receivers have the same (smallest) batch ID). These results reveal that the $LR$ policy may also be throughput optimal, in the case of transmitting a file of infinite size or a stream of packets.

\begin{figure}
\centering
\includegraphics[scale = 0.5, trim = 0.5cm 0cm 0cm 0cm, clip]{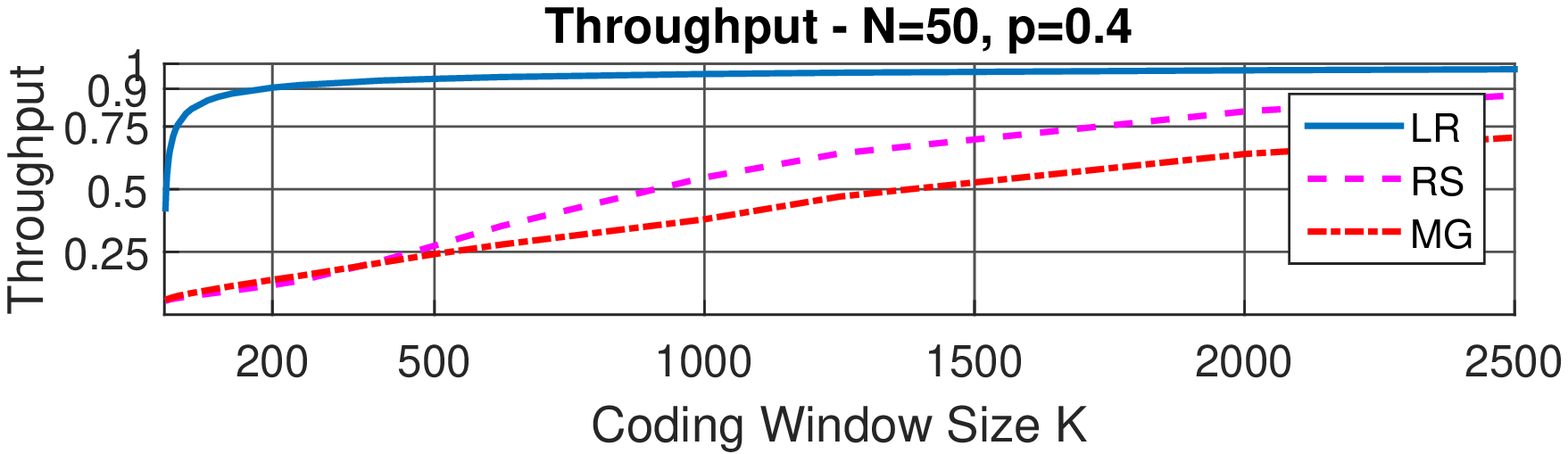}
\caption{Normalized throughput comparison.}
\label{thr1}
\end{figure}

\begin{figure}
\centering
\includegraphics[scale = 0.5, trim = 0.5cm 0cm 0cm 0cm, clip]{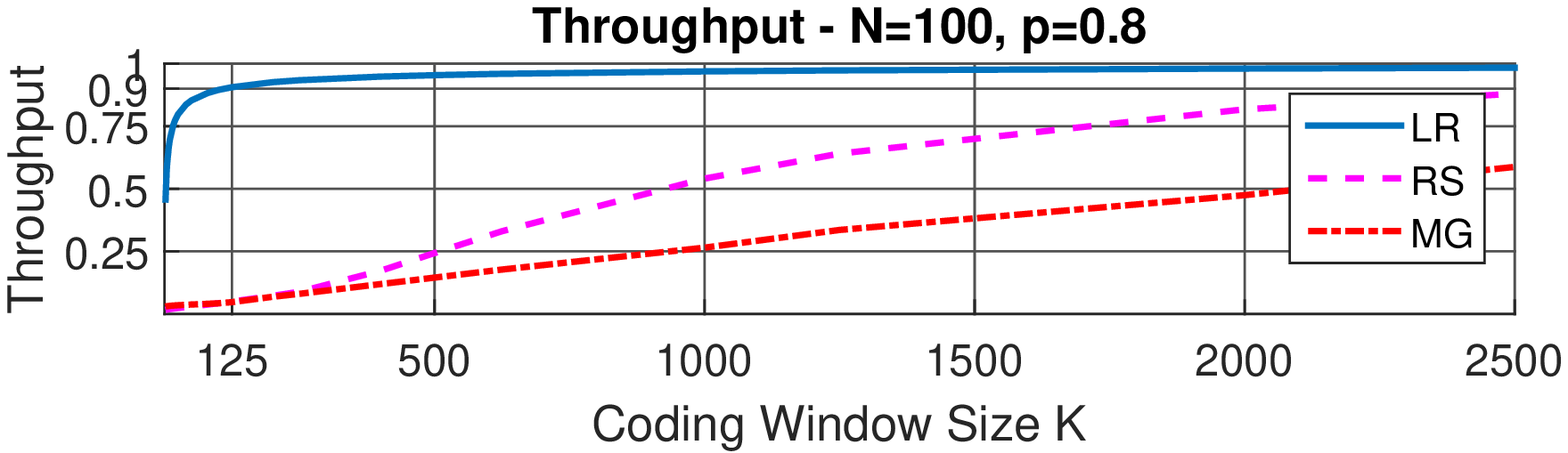}
\caption{Normalized throughput comparison.}
\label{thr2}
\end{figure}
 
\subsection{Accuracy of Approximations}
In this section we will evaluate the accuracy of the approximations that are presented in section \ref{approx}. We compared eq. \ref{Alpha} with the actual value of $A(N)$ (eq. \ref{ETFK1}). Since $A(N)$ solely depends on the number of receivers, we compared those values for a wide range of receivers; from $2$ up to $5000$, by calculating the percent of the absolute difference. We concluded that eq. \ref{Alpha} represents the value $A(N)$ reasonably accurate since the average percent difference was found to be $0.081\%$ with a maximum value of $4\%$.

In \cite{skevakis2016decoding}, we established the accuracy of eq. \ref{ETFK1} with respect to both the actual and the experimental file transfer completion time. Figure \ref{app} compares the file transfer completion time (normalized by $F/p$) based on eq. \ref{ETFK1} and \ref{ETFKappr}\footnote{The values of $K$ start from 8 as a result of the constraint $K > \widetilde{n}^2(1-p)$ for eq. \ref{ETFK1}.}. As we can see, our approximation accurately represent the completion time. The mean and maximum errors of our approximation were calculated to be around $0.13\%$ and $0.26\%$, respectively (up to $0.22\%$ and $0.33\%$ from the entirety of our experiments).

\begin{figure}
\centering
\includegraphics[scale = 0.5, trim = 1cm 0cm 0cm 0cm, clip]{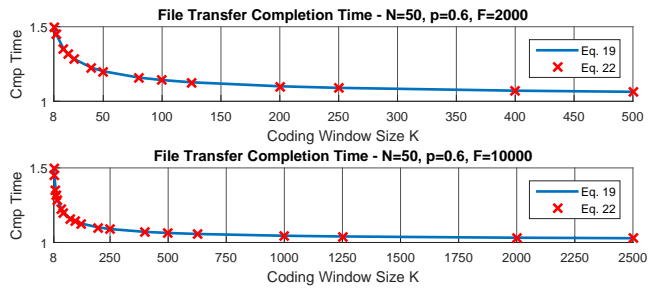}
\caption{Accuracy of eq. \ref{ETFKappr}}
\label{app}
\end{figure}

In table \ref{table1}, we show the minimum coding window size that achieves a file transfer completion time of at most $\epsilon$ times more than the optimal (when $K=F$) for different values of the file size $F$. The minimum $K$ is expressed as a percentage of $F$. The first column ($LR$) is derived from the experiments and the second and third column are derived from eq. \ref{epsilon} and \ref{epsilonappr}, respectively. We removed the restriction that $F/K$ must be an integer for the last 2 columns in order to evaluate the accuracy of our approximation. As we can see, eq. \ref{epsilonappr} produces almost the same results as eq. \ref{epsilon}; the differences are most of the time less than $0.5\%$. One can notice that the theoretical minimum $K$ is the same as the experimental one in all of the cases. Furthermore, this table verifies that near-optimal file transfer completion time can be achieved with a much smaller coding window size; up to $100$ times smaller for $\epsilon=10\%$ and up to $4$ times smaller for $\epsilon=1\%$.

\begin{table}
\large
\centering
\adjustbox{max width=\columnwidth}{
\begin{tabular}{| c || ccc | ccc | ccc |}
\hline
   $\epsilon$ & \multicolumn{3}{c|}{LR} & \multicolumn{3}{c|}{Eq (\ref{epsilon})} & \multicolumn{3}{c|}{Eq. (\ref{epsilonappr})}\\ 
   \hline \hline
   & \multicolumn{3}{c|}{$F$}   & \multicolumn{3}{c|}{$F$}   & \multicolumn{3}{c|}{$F$}\\
   & $2K$ & $5K$&$10K$& $2K$ & $5K$&$10K$& $2K$ & $5K$&$10K$ \\ \hline
  $10\%$  & 4\% & 2\% & 1\%  & 3.3\% & 1.52\%  & 0.83\% & 3.35\%   &1.54\% & 0.83\%  \\ \hline
  $1\%$   & 50\% & 50\%  & 25\% & 47.3\%     & 34.12\% & 24.93\%    &47.75\% & 34.3\%    &24.98\% \\ \hline
\end{tabular}}
\caption{Percentage of minimum Coding Window Size - $p = 0.8$, $N = 50$}
\label{table1}
\end{table}


\subsection{Limited Feedback} \label{limited}

The assumption of complete feedback information is impractical in realistic systems, especially as the number of receivers increases. In this section we will present experimental results of a modification of the $LR$ policy for systems with limited feedback. We assume that each receiver transmits an ACK (which is received by the base station instantly and without errors) when the receiver receives all of the encoded packets that are needed to decode a single batch. Therefore, the base station has knowledge of the batch ID of each receiver but not of that of each \textit{connected} receiver. The modified $LR$ transmits an encoded packet of batch $i$, where $i$ is the minimum batch ID among \textit{all} of the receivers. We note here that the approximations of section \ref{approx} and the results of \cite{skevakis2016decoding} model the behaviour of the modified $LR$. However, as we showed in \cite{skevakis2016decoding} they accurately represent the behaviour of the original $LR$ policy. Figure \ref{lim} depicts the normalized file transfer completion time under the modified $LR$ policy and the "original" $LR$. As we can see, the modified $LR$ is able to achieve near optimal completion time with significantly less feedback requirements. The percent of the difference in the completion time under the two policies (in the scenario of figure \ref{lim}) was calculated to have an average value of $2.3\%$ and a maximum value of $8\%$. The difference between the policies is decreasing as we increase either $N$, $F$ or $p$.

\begin{figure}
\centering
\includegraphics[scale = 0.5, trim = 1cm 0cm 0cm 0cm, clip]{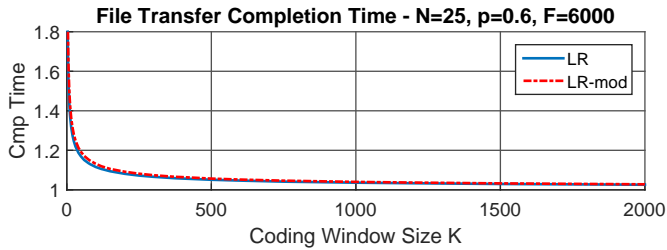}
\caption{Modified $LR$ vs $LR$}
\label{lim}
\end{figure}

\section{Conclusions}\label{conclusions}

We apply chunked RLNC in a single-hop network for broadcast communications, where a single file is transmitted to an arbitrary number of receivers through unreliable channels. In a previous work we proposed and evaluated a scheduling policy for chunked RLNC, namely the Least Received (LR). We derived a closed form formula for the expected file transfer completion time, under the $LR$ policy and one for the minimum coding window size so that the completion time is upper bounded by a user defined delay constraint.

In this work, we proved the optimality of the $LR$ policy regarding the expected file transfer completion time; i.e. there is no feasible policy can achieve lower file transfer completion time, for a given number of receivers, file size, coding window size and erasure probability of the channels. Secondly, we derived simple and accurate approximations for the formulas presented in \cite{skevakis2016decoding}. From or simulation results we also hinted that the $LR$ policy may be throughput optimal. This however is beyond the scope of this work. Finally, we proposed a modification of the $LR$ policy in the case of limited feedback from the receivers. We showed that this modification can achieve almost the same completion time as the $LR$ with substantially less feedback.

Our future research will focus on developing a policy, based on the $LR$, in the case of minimal feedback; the receivers will only acknowledge the reception of the entire file. Moreover, we will focus on expanding our system model for multicast communications.

\appendices
\section{Proof of Lemma \ref{lemmaA}}
Since lemma \ref{lemmaA} is independent of the policy, we will drop the policy superscript. Before we introduce the proof, we remind the reader that, at time $t > t_s$, we denote with $i$ ($j$) the receivers of the set $R_{i}(t_s)$ ($R_{j}(t_s)$) that remain in the set $R_{\beta(t_s)}(t)$ at $t$. 

Let $C_k$ denote the binary random variable representing the connectivity of receiver $r_k$ (1 for connected and 0 otherwise) and $C(t)$ the random vector of the connectivities of all of the receivers at $t$.
For any sample path $\omega$ for which $\min\limits_{i} X_{i}(\omega, t'+1) > \min\limits_{j} X_{j}(\omega, t'+1)$ at some $t' > t_s$, there must have been at least one time slot $t$ $(t_s < t < t')$ where $\min\limits_{i} X_{i}(\omega, t) = \min\limits_{j} X_{j}(\omega, t)$. At each such $t$, if there exists a pair ($i$,$j$) such that $r_{i} \in R_{i}(t_s) \cap R_{\beta(t_s)}(t), r_{j} \in R_{j}(t_s) \cap  R_{\beta(t_s)}(t)$ and the following conditions are met:
\begin{itemize}
\item $X_{i}(t) = \min\limits_{i} X_{i}(t), C_{i}(t) = 1$
\item $X_{j}(t) = \min\limits_{i} X_{i}(t), C_{j}(t) = 0$
\end{itemize}
we will switch to an \textit{equivalent} sample path $\omega'$ with the following property :
\begin{center}
$C'_{i}(t) = 0$ and $C'_{j}(t) = 1$
\end{center}

\noindent Note: we only need to find one pair of ($i$,$j$) in order to guarantee that $\min\limits_{i} X_{i}^{\pi}(\omega', t+1) \leq \min\limits_{j} X_{j}^{\pi}(\omega', t+1)$. For any sample path $\omega$ (as described above), at least one $t$ with the above mentioned property will surely exist.

\begin{figure}
\centering
\includegraphics[scale = 0.375]{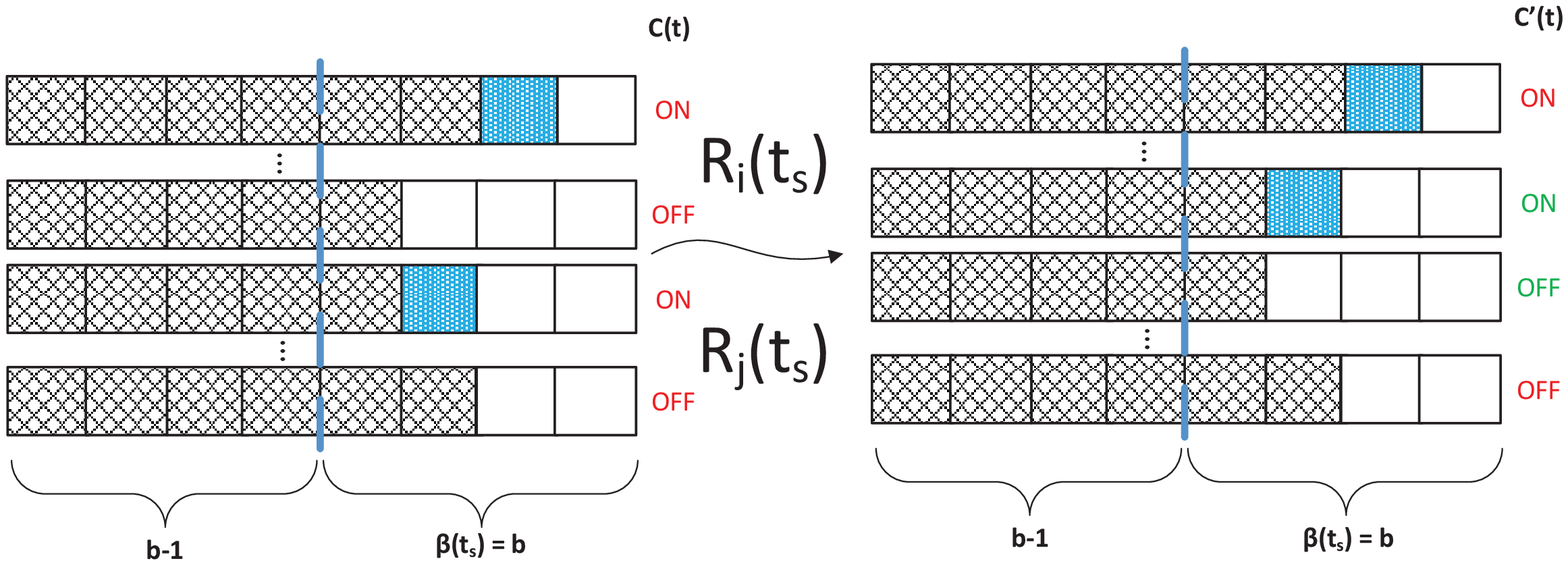}
\caption{Example of the connectivity coupling.}
\label{state111}
\end{figure}

An example of the sample path switching (in order to guarantee stochastic dominance), $\omega|C(t) \rightsquigarrow \omega'|C'(t)$, can be seen in figure \ref{state111}. Since the channels between the base station and the receivers are assumed to be identical and independent across time and receivers, the resulting r.v. $X(\omega')$ has the same probability distribution as $X(\omega)$ (since the stochastic process $\boldsymbol{C'}$ has the same distribution as $\boldsymbol{C}$). By switching to the sample path $\omega'$ at each $t$ (as described above), we can safely state that:
\begin{center}
$\min\limits_{i} X_{i}(t+1) \leq \min\limits_{j} X_{j}(t+1) \hspace{2 em} t \geq t_s,$\\$ \forall i,j : r_{i} \in R_{i}(t_s) \cap R_{\beta(t_s)}(t), r_{j} \in R_{j}(t_s) \cap R_{\beta(t_s)}(t)$ 
\end{center}


\section{Proof of Lemma \ref{lemma1}}

From the fact that ${t_{b-1}^{ON}}^{(\widetilde{\pi})} \leq {t_{b-1}^{ON}}^{(\pi)}$ and by the definition of those time slots we know that :
\begin{itemize}
\item At ${t_{b-1}^{ON}}^{(\pi)}$, under $\widetilde{\pi}$, the bottleneck receivers will be \textit{at least} at the end of batch $b-1$.
\item At ${t_{b-1}^{ON}}^{(\pi)}$, under $\pi$, the bottleneck receivers will be \textit{exactly} at the end of batch $b-1$.
\end{itemize}

Let us denote with $i^{(\widetilde{\pi})}$ the $i$'s such that $r_i \in {R_{b-1}^*}^{(\widetilde{\pi})}({t_{b-1}^{ON}}^{(\pi)})$ and with $i^{(\pi)}$ the $i$'s such that $r_i \in {R_{b-1}^*}^{(\pi)}({t_{b-1}^{ON}}^{(\pi)})$. We will distinguish two cases that the conditions of this lemma can hold depending on the cardinality of the set ${R_{b-1}^*}^{(\widetilde{\pi})}({t_{b-1}^{ON}}^{(\pi)})$ (whether its empty or not). 
\vspace{1 em}

\noindent \textit{Case 1: ${R_{b-1}^*}^{(\widetilde{\pi})}({t_{b-1}^{ON}}^{(\pi)}) \neq \emptyset$.}\footnote{This can occur only when ${t_{b-1}^{ON}}^{(\widetilde{\pi})} = {t_{b-1}^{ON}}^{(\pi)}$}

\begin{figure}[H]
\centering
\includegraphics[scale=.3]{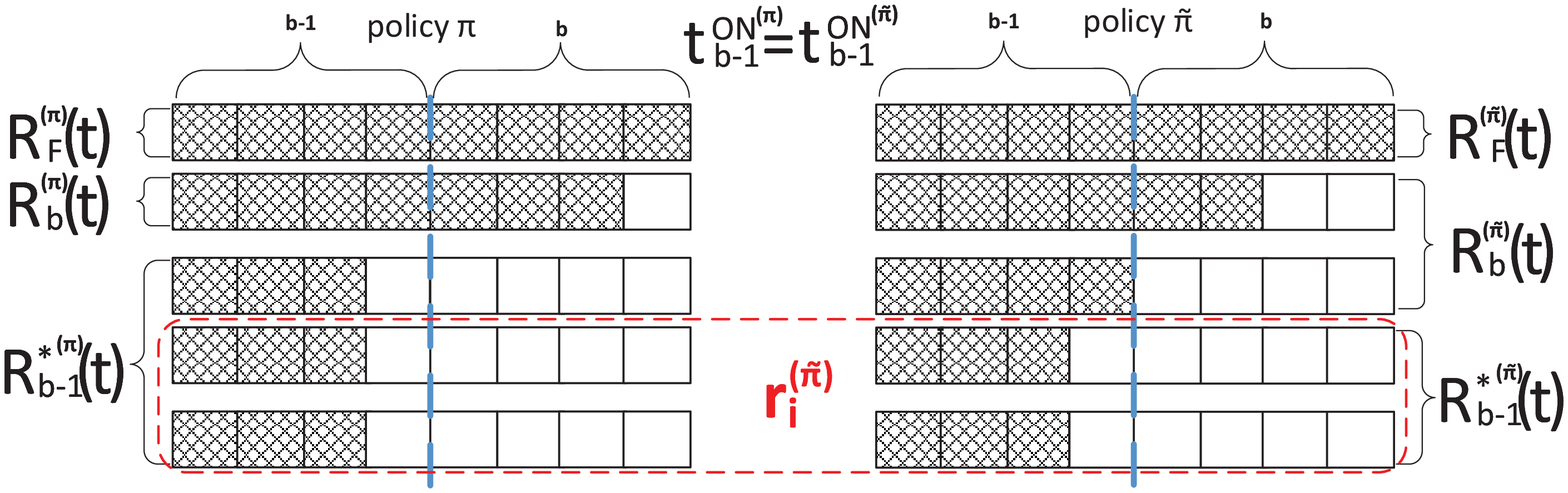}
\caption{System at ${t_{b-1}^{ON}}^{(\pi)}$ under $\pi$ and $\widetilde{\pi}$}
\label{AppA_1}
\end{figure}

Figure \ref{AppA_1} depicts the system at time ${t_{b-1}^{ON}}^{(\pi)}$ under both policies. At that time slot, both $\pi$ and $\widetilde{\pi}$ will choose batch $b-1$ for transmission. Thus,
\begin{center}
$X_{i^{(\widetilde{\pi})}}^{(\widetilde{\pi})}({t_{b-1}^{ON}}^{(\pi)} + 1) = X_{i^{(\widetilde{\pi})}}^{(\pi)}({t_{b-1}^{ON}}^{(\pi)} + 1) $\\$ \forall {i^{(\widetilde{\pi})}} : r_{i^{(\widetilde{\pi})}} \in {R_{b-1}^*}^{(\widetilde{\pi})}({t_{b-1}^{ON}}^{(\pi)})$.
\end{center}

It's easy to see that, since no conflict slots occur under any policy, the above equation will hold for the remaining time slots. Therefore, 
\begin{center}
$T_{i^{(\widetilde{\pi})}}^{(\widetilde{\pi})} = T_{i^{(\widetilde{\pi})}}^{(\pi)}$
\end{center}

By applying lemma \ref{lemmaA} twice, one for each policy, with the following parameters (where $\pi'$ refers to either $\pi$ or $\widetilde{\pi}$): 
\begin{itemize}
\item $t_s = {t_{b-1}^{ON}}^{(\pi)} + 1$.
\item $R_{\beta(t_s)}^{\pi'}(t_s) = R_b^{(\pi')}({t_{b-1}^{ON}}^{(\pi)})$
\item $R_{i}^{\pi'}(t_s) = {R_{b-1}^*}^{(\pi')}({t_{b-1}^{ON}}^{(\pi)})$. 
\item $R_{j}^{\pi'}(t_s) = {R_{b}}^{(\pi')}({t_{b-1}^{ON}}^{(\pi)})$. 
\end{itemize}

we can see that 
\begin{itemize}

\item $\max\limits_{i^{(\widetilde{\pi})}}T_{i^{(\widetilde{\pi})}}^{\widetilde{\pi}} \geq \max\limits_{j^{(\widetilde{\pi})}}T_{j^{(\widetilde{\pi})}}^{\widetilde{\pi}} \hspace{1 em} {j^{(\widetilde{\pi})}} : r_{j^{(\widetilde{\pi})}} \in {R_{b}}^{(\widetilde{\pi})}({t_{b-1}^{ON}}^{(\pi)})$.

\item $\max\limits_{i^{({\pi})}}T_{i^{({\pi})}}^{{\pi}} \geq \max\limits_{j^{({\pi})}}T_{j^{({\pi})}}^{{\pi}} \hspace{1 em} {j^{({\pi})}} : r_{j^{({\pi})}} \in {R_{b}}^{({\pi})}({t_{b-1}^{ON}}^{(\pi)})$.

\end{itemize}

If there exists a receiver $r_k$ such that $r_k \in {R_{b-1}^*}^{(\pi)}({t_{b-1}^{ON}}^{(\pi)})$ and $r_k \notin{R_{b-1}^*}^{(\widetilde{\pi})}({t_{b-1}^{ON}}^{(\pi)})$, then $X_{k}^{(\pi)}(t) < X_{k}^{(\widetilde{\pi})}(t)$ and thus $T_k^{(\pi)} > T_k^{(\widetilde{\pi})}$. Hence, $\max\limits_{i^{({\pi})}}T_{i^{({\pi})}}^{{\pi}} > \max\limits_{i^{(\widetilde{\pi})}}T_{i^{(\widetilde{\pi})}}^{\widetilde{\pi}}$.\\

\noindent If not, then  $\max\limits_{i^{({\pi})}}T_{i^{({\pi})}}^{{\pi}} = \max\limits_{i^{(\widetilde{\pi})}}T_{i^{(\widetilde{\pi})}}^{\widetilde{\pi}}$.

\noindent Therefore, we can see that :

\begin{center}
$T^{\widetilde{\pi}} = \max\{\max\limits_{i^{(\widetilde{\pi})}}T_{i^{(\widetilde{\pi})}}^{\widetilde{\pi}}, \max\limits_{j^{(\widetilde{\pi})}}T_{j^{(\widetilde{\pi})}}^{\widetilde{\pi}} \} \stackrel{}{=} \max\limits_{i^{(\widetilde{\pi})}}T_{i^{(\widetilde{\pi})}}^{\widetilde{\pi}} \stackrel{}{\leq}$ \\
$\max\limits_{i^{(\pi)}}T_{i^{(\pi)}}^{\pi} \leq \max\{\max\limits_{i^{(\pi)}}T_{i^{(\pi)}}^{\pi}, \max\limits_{j^{(\pi)}}T_{j^{(\pi)}}^{\pi} \} = T^{\pi} \Rightarrow$\\$ T^{\widetilde{\pi}} \leq T^{\pi}$

\end{center}

\noindent \textit{Case 2: ${R_{b-1}^*}^{(\widetilde{\pi})}({t_{b-1}^{ON}}^{(\pi)}) = \emptyset$.}
\begin{figure}[H]
\centering
\includegraphics[scale=.33]{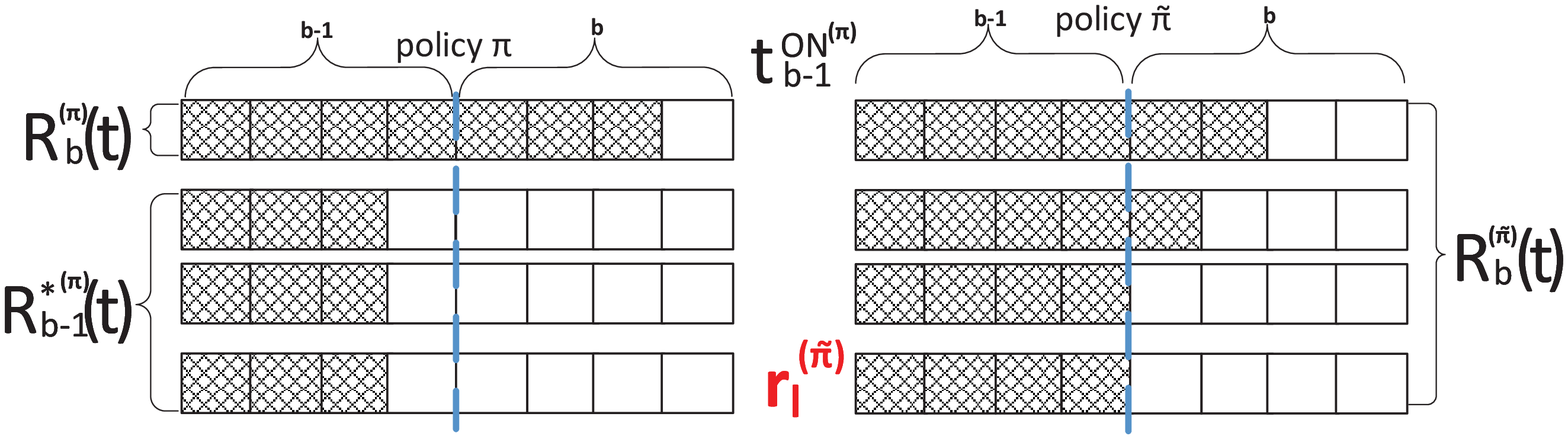}
\caption{System at ${t_{b-1}^{ON}}^{(\pi)}$ under $\pi$ and $\widetilde{\pi}$}
\label{AppA_2}
\end{figure}

In this case, under $\widetilde{\pi}$ at ${t_{b-1}^{ON}}^{(\pi)}$, no receiver is in batch $b-1$. Let $r_l^{(\widetilde{\pi})}$ denote one of the bottleneck receivers of $\widetilde{\pi}$ at that time slot. Without loss of generality, we assume that $r_l^{(\widetilde{\pi})}$ belongs in ${R_{b-1}^*}^{(\pi)}({t_{b-1}^{ON}}^{(\pi)})$.
Hence,
\begin{center}
$X_{l^{(\widetilde{\pi})}}^{(\widetilde{\pi})}({t_{b-1}^{ON}}^{(\pi)}) > X_{l^{(\widetilde{\pi})}}^{(\pi)}({t_{b-1}^{ON}}^{(\pi)})$
\end{center} 
Since, under $\widetilde{\pi}$, no more conflict slots will occur,
\begin{center}
$X_{l^{(\widetilde{\pi})}}^{(\widetilde{\pi})}(t) > X_{l^{(\widetilde{\pi})}}^{(\pi)}(t), \hspace{1.5 em} {t_{b-1}^{ON}}^{(\pi)} \leq t \leq \min\{T_{l^{(\widetilde{\pi})}}^{\widetilde{\pi}}, T_{l^{(\widetilde{\pi})}}^{\pi}\}$
\end{center} 
%
Hence, $T_{l^{(\widetilde{\pi})}}^{\widetilde{\pi}} < T_{l^{(\widetilde{\pi})}}^{\pi}$. 

By using lemma \ref{lemmaA} for $t_s = {t_{b-1}^{ON}}^{(\pi)} + 1$ and for  $R_{i}^{\pi}(t_s) = \{r_{l}^{(\widetilde{\pi})}\}$ and $R_{j}^{\pi}(t_s) = R_b^{(\widetilde{\pi})}({t_{b-1}^{ON}}^{(\pi)} + 1)\setminus \{r_{l}^{(\widetilde{\pi})}\}$, we can show that $T_{l^{(\widetilde{\pi})}}^{\widetilde{\pi}} \geq T_{k}^{\widetilde{\pi}}$, where ${r_k}$ is any receiver except $r_{l^{(\widetilde{\pi})}}$.
Hence, 
\begin{center}
$T^{\widetilde{\pi}} = \max\{T_{l^{(\widetilde{\pi})}}^{\widetilde{\pi}}, \max\limits_{k}T_{k}^{\widetilde{\pi}} \} = T_{l^{(\widetilde{\pi})}}^{\widetilde{\pi}}$ and \\
$T^{{\pi}} = \max\{T_{l^{(\widetilde{\pi})}}^{{\pi}}, \max\limits_{k}T_{k}^{{\pi}} \} \geq T_{l^{(\widetilde{\pi})}}^{{\pi}} >  T_{l^{(\widetilde{\pi})}}^{\widetilde{\pi}}$
\end{center}
Therefore,
\begin{center}
$T^{\widetilde{\pi}} < T^{\pi}$.
\end{center}

\section{Proof of Lemma \ref{lemma2}}

\noindent From the fact that ${t_{b-1}}^{(\widetilde{\pi})} \leq {t_{b-1}}^{(\pi)}$ and by the definition of the time slot ${t_{b-1}}^{(\pi)}$ we know that :
\begin{itemize}
\item At ${t_{b-1}}^{(\pi)}$, under $\widetilde{\pi}$, the bottleneck receivers will be \textit{at least} at the end of batch $b-1$.
\item At ${t_{b-1}}^{(\pi)}$, under $\pi$, the bottleneck receivers will be \textit{exactly} at the end of batch $b-1$.
\end{itemize}

\noindent Let us denote with $i^{(\widetilde{\pi})}$ the $i$'s such that $r_i \in {R_{b-1}^*}^{(\widetilde{\pi})}({t_{b-1}}^{(\pi)})$ and with $i^{(\pi)}$ the $i$'s such that $r_i \in {R_{b-1}^*}^{(\pi)}({t_{b-1}}^{(\pi)})$. As in the previous lemma (lemma \ref{lemma1}) we will distinguish between two cases.

\noindent \textit{Case 1 :} ${R_{b-1}^*}^{(\widetilde{\pi})}({t_{b-1}}^{(\pi)}) \neq \emptyset$.\footnote{In this case, we need condition 2 ( ${R_{b-1}^*}^{(\widetilde{\pi})}({t_{b-1}}^{(\pi)}) \subseteq {R^*_{b-1}}^{(\pi)}({t_{b-1}}^{(\pi)})$) in order to guarantee that all of the the $i^{(\widetilde{\pi})}$'s are included in the $i^{({\pi})}$'s.}\\

In section \ref{t_b-1} we derived the optimal decisions from $t_{b-1}$ and onwards. Based on that, $\widetilde{\pi}$ will transmit a packet to every $r_{i^{(\widetilde{\pi})}}$ that it is ON. Thus, 
\begin{center}
$X_{i^{(\widetilde{\pi})}}^{\widetilde{\pi}}(t+1) \geq X_{i^{(\widetilde{\pi})}}^{\pi}(t+1) \hspace{2 em} t \geq {t_{b-1}}^{(\pi)}, \forall r_{i^{(\widetilde{\pi})}} \in {R_{b-1}^*}^{(\widetilde{\pi})}(t)$\footnote{It is evident that all the $r_i$'s that belong to ${R_{b-1}^*}^{(\pi)}(t)$ and not to ${R_{b-1}^*}^{(\widetilde{\pi})}(t)$,  will have received more packets with $\widetilde{\pi}$ than with $\pi$.}
\end{center}
Each of the $r_{i^{(\widetilde{\pi})}}$'s will move to the set $R_{b}^{(\widetilde{\pi})}$ once they receive one packet and no new receivers can enter either set, thus :
\begin{center}
${R_{b-1}^*}^{(\widetilde{\pi})}(t+1) \subseteq {R_{b-1}^*}^{(\pi)}(t+1)\footnote{Obviously, the fact that we focus on the $i^{(\widetilde{\pi})}$'s and disregard some of the $i^{(\pi)}$'s does not affect this result since ${R_{b-1}^{*}}^{(\widetilde{\pi})}({t_{b-1}}^{(\pi)}) \subseteq {R_{b-1}^*}^{(\pi)}({t_{b-1}}^{(\pi)})$} \hspace{2 em} t \geq {t_{b-1}}^{(\pi)}$
\end{center}
Therefore, as in section \ref{t_b-1}, ${t_{b-1}^{ON}}^{(\widetilde{\pi})} \leq {t_{b-1}^{ON}}^{(\pi)}$ and thus ${R_{b-1}^*}^{(\widetilde{\pi})}({t_{b-1}^{ON}}^{(\pi)}) \subseteq {R_{b-1}^*}^{(\pi)}({t_{b-1}^{ON}}^{\pi)})$. By using Lemma \ref{lemma1} we can see that $T^{\widetilde{\pi}} \leq T^{\pi}$.
\vspace{1 em}

\noindent \textit{Case 2 :} ${R_{b-1}^*}^{(\widetilde{\pi})}({t_{b-1}}^{(\pi)}) = \emptyset$.

Let us denote with $r_{l^{(\widetilde{\pi})}}$ and $r_{l^{(\pi)}}$ one of the bottleneck receiver under $\widetilde{\pi}$ and $\pi$ at ${t_{b-1}}^{(\pi)}$, respectively. Then,
\begin{center}
$X_{l^{(\widetilde{\pi})}}({t_{b-1}}^{(\pi)}) < X_{l^{(\pi)}}({t_{b-1}}^{(\pi)})$
\end{center}
Since there will be no conflict slots with $\widetilde{\pi}$ (since no receiver is in $R_{b-1}$), $r_{l^{(\widetilde{\pi})}}$ will receive a packet at every $t$ that it is ON, whereas $r_{l^{(\pi)}}$ might receive a packet at every $t$ that it is ON. Therefore, 
\begin{equation}
T_{{l^{(\widetilde{\pi})}}}^{\widetilde{\pi}} < T_{l^{(\pi)}}^{\pi}
\label{17}
\end{equation}
By using lemma \ref{lemmaA}, we can see that :
\begin{center}
$T_{l^{(\widetilde{\pi})}}^{\widetilde{\pi}} \geq \min\limits_k T_k^{\widetilde{\pi}}$ , where $r_k$ is any receiver.
\end{center}
\begin{center}

Hence, 
$T^{\widetilde{\pi}} = T_{r_{l^{(\widetilde{\pi})}}}^{\widetilde{\pi}}$.

\end{center}
\noindent Using eq. \ref{17}, we can see that : 
\begin{center}
$T^{\pi} \geq T_{l^{(\pi)}}^{\pi} > T_{{l^{(\widetilde{\pi})}}}^{\widetilde{\pi}} = T^{\widetilde{\pi}}$
\end{center}

\bibliography{references}
\bibliographystyle{IEEETran}

\end{document}